\documentclass[10pt,article]{IEEEtran}
 \usepackage{epsfig,latexsym,amsmath,amssymb,setspace,cite,color,graphicx,algorithm,algorithmic,cases}
\usepackage[caption=false,font=footnotesize]{subfig}
 \usepackage{booktabs}

\usepackage{array}\usepackage{makecell}
\usepackage{pict2e} 
\usepackage{amsmath}
\usepackage{amsfonts}
\usepackage{amssymb}
\usepackage{dsfont}
\usepackage{bm}
\usepackage{amsthm}
\usepackage{newlfont}
\usepackage{float}
\usepackage{hyperref}
\usepackage{algorithm}
\usepackage{algorithmic}
\usepackage{enumerate}
\usepackage{chngcntr}
\usepackage{mathtools}
 \usepackage{array,multirow,graphicx}
 \usepackage{float}
\usepackage{breqn}
\usepackage{stmaryrd}
\usepackage{cite}
\usepackage[yyyymmdd]{datetime}
\hypersetup{
    colorlinks=true,  
    linkcolor= blue,  
    citecolor=blue 
}

\newcolumntype{C}[1]{>{\centering\let\newline\\\arraybackslash\hspace{0pt}}m{#1}}

\DeclareMathOperator*{\concat}{\lVert}

\begin{document}
\sloppy
\allowdisplaybreaks[0]

\newtheorem{thm}{Theorem} 
\newtheorem{lem}{Lemma}
\newtheorem{prop}{Proposition}
\newtheorem{cor}{Corollary}
\newtheorem{defn}{Definition}
\newcommand{\remarkend}{\IEEEQEDopen}
\newtheorem{remark}{Remark}
\newtheorem{rem}{Remark}
\newtheorem{ex}{Example}
\newtheorem{pro}{Property}

\newenvironment{example}[1][Example]{\begin{trivlist}
\item[\hskip \labelsep {\bfseries #1}]}{\end{trivlist}}

\renewcommand{\qedsymbol}{ \begin{tiny}$\blacksquare$ \end{tiny} }

\renewcommand{\leq}{\leqslant}
\renewcommand{\geq}{\geqslant}

{\title {Explicit Wiretap Channel Codes via Source Coding, Universal Hashing, and Distribution Approximation, When the Channels' Statistics are Uncertain}}

\author{\IEEEauthorblockN{R\'emi A. Chou}\\
 
\thanks{
Part of this work has been presented at the 2018 IEEE International Symposium on Information Theory (ISIT) \cite{chou2018b}. This work was supported in part by NSF grants CCF-1850227 and CCF-2047913. 
}
 }

\maketitle
\begin{abstract}
We consider wiretap channels with uncertainty on the eavesdropper channel under (i)  noisy blockwise type~II, (ii)~compound, or (iii) arbitrarily varying models. We present explicit wiretap codes that can handle these models in a unified manner and only rely on three primitives, namely source coding with side information, universal hashing, and distribution approximation. Our explicit wiretap codes achieve the  best known single-letter achievable rates, previously obtained non-constructively, for the models considered.
Our results are obtained for strong secrecy, do not require a pre-shared secret between the legitimate users, and do not require any symmetry properties on the channel. An extension of our results to  compound main channels is also~derived via  new capacity-achieving polar coding schemes for compound settings.  
\end{abstract} 

\section{Introduction}
The wiretap channel \cite{Wyner75} is a fundamental primitive to model eavesdropping at the physical layer~\cite{liang2009information,Bloch11}. Beyond theoretical results that characterize the secrecy capacity for this model, significant progress has been made in the development of explicit wiretap codes for Wyner's wiretap channel~\cite{Wyner75}. Specifically, coding schemes based on low-density parity-check codes, e.g.,~\cite{Thangaraj2007,Subramanian2011,Rathi2011}, polar codes, e.g.~\cite{Mahdavifar11,Sasoglu13,Andersson2013,Wei14,renes2013efficient,Gulcu14,chou2015polaritw}, and invertible extractors, e.g.,~\cite{hayashi2010construction,Bellare2012,tyagi2015universal}, have been successfully developed for Wyner's model~\cite{Wyner75} or some of its special cases.

An assumption  made by all the above references is that the eavesdropper channel statistics are perfectly known by the legitimate users. 
To model uncertainty, several models have been introduced: Type~II models \cite{ozarow1984wire,goldfeld2016semantic,nafea2016new}, where the eavesdropper can learn an arbitrary and unknown part of the legitimate sender codeword, and models where the eavesdropper channel statistics are not perfectly known but only known to belong to a given set. These latter models are useful when the physical location of the eavesdropper is uncertain from the point of view of the legitimate users, and include compound models~\cite{Liang09,bjelakovic2013secrecy}, where the channel statistics are  known to be fixed for all channel uses, and arbitrarily varying models~\cite{molavianjazi2009arbitrary,bjelakovic2013capacity}, where the channel statistics change at each channel use.

Our contributions are summarized as follows. (i) We construct explicit wiretap codes that 
achieve the  best known single-letter achievable rates, previously obtained non-constructively,  when uncertainty holds on the eavesdropper channel  under a  noisy~blockwise type~II,  compound, or~arbitrarily varying model. (ii) We prove the sufficiency of three primitives to construct such wiretap codes: source coding with side information, universal hashing, and distribution approximation. (iii) We extend our results to the case where uncertainty holds on the main channel according to a compound model. (iv)~We demonstrate that all the  models considered in this paper can be handled in a unified manner by the same encoding and decoding schemes, up to an appropriate choice of parameters.   We stress that our results are obtained for strong secrecy, do not require a pre-shared secret between the legitimate users, and do not require any symmetry properties on the channel.

 Our approach consists in separately handling the reliability constraint and the security constraints. The reliability constraint is handled via a combination of source coding with side information and distribution approximation implemented with polar codes. The security constraints are handled with a combination of universal hashing and distribution approximation implemented via two-universal hash functions and polar codes, respectively. 
 The main difficulty in our approach is to combine universal hashing and source coding with side information such that (i) non-symmetric and non-degraded channels can be handled, and (ii) the analysis of the security of the overall coding scheme is possible. 
 (i) is performed via the idea of block-Markov coding as introduced in \cite{hassani2014universal,Mondelli14} with   the following two important modifications to enable (ii): (1)~Each encoding block of the block-Markov construction is constructed from $L$ sub-blocks in which all the involved random variables have the same joint distribution across all sub-blocks. (2)~The construction of each encoding block is such that the encoder output distribution  approaches a fixed target distribution. In particular, these two points are key to analyzing the security of universal hashing via the leftover hash lemma \cite{renner2008security}, whose application in our analysis raises several additional challenges. First, while the leftover hash lemma proves a security guarantee
on the output of a hash function, in our coding scheme, we need to prove a security
guarantee on a message $M$ that is not obtained as the output of a hash function. To circumvent this difficulty, we
prove the statistical equivalence between our coding scheme and another coding scheme where the
message $M$ is obtained as the output of a hash function. Second, because of the block-Markov construction, a precise study of the inter-dependencies between the encoding blocks is needed to evaluate the overall leakage when considering all the blocks jointly.

In Section \ref{sec:main}, we formally describe the model considered in this paper. In Section \ref{sec:results}, we state our main results. In Section~\ref{sec:preshared}, we describe our proposed coding scheme.  The analysis of our coding scheme is presented in Sections \ref{sec:proof1a}, \ref{sec:initphase}, \ref{sec:proof2}. In Section \ref{secext}, we present an extension of our results to the case where uncertainty holds on the legitimate user channel under a compound model \cite{liang2009information,bjelakovic2013secrecy}. Finally, in Section \ref{sec:concl}, we provide concluding remarks.

\section{Notation} \label{sec:notations}  For $a,b \in \mathbb{R}_+$, define $\llbracket a, b\rrbracket \triangleq [\lfloor a \rfloor, \lceil b \rceil] \cap \mathbb{N}$. The components of a vector $X^{1:N}$ of size $N$ are denoted with superscripts, i.e., $X^{1:N} \triangleq (X^i)_{i \in \llbracket  1 ,N \rrbracket}$. For any set $\mathcal{A} \subseteq \llbracket 1,N \rrbracket$, let $X^{1:N}[\mathcal{A}]$ be the components of $X^{1:N}$ whose indices are in $\mathcal{A}$. For two distributions $p_{XY}$ and $q_{XY}$ defined over $\mathcal{X} \times \mathcal{Y}$, define  the variational distance between $p_X$ and $q_X$ as $\mathbb{V}(p_X, q_X)\triangleq \sum_{x\in \mathcal{X}} |p_X(x) - q_X(x)|$, the Kullback-Leibler divergence between $p_X$ and $q_X$ as $\mathbb{D}(p_X \lVert q_X)$, and the conditional Kullback-Leibler divergence between $p_{Y|X}$ and $q_{Y|X}$ as 
$
    \mathbb{E}_{p_X}[\mathbb{D}(p_{Y \vert X}\Vert q_{Y \vert X} )] \triangleq  \sum _{x \in \mathcal{X}}p_X(x) \mathbb{D}(p_{Y \vert X=x}\Vert q_{Y \vert X=x}).
$
Unless otherwise specified, capital letters denote random variables, whereas lowercase letters designate realizations of associated random variables, e.g., $x$ is a realization of the random variable $X$. Let $\mathds{1}\{ \omega \}$ be the indicator function, which is equal to $1$ if the predicate $\omega$ is true and $0$ otherwise. For any $x\in \mathbb{R}$, define $[x]^+ \triangleq \max(0,x)$.  Finally, GF$(2^N)$ denotes a finite field of order $2^N$.
\section{Model and known results}  \label{sec:main}
Consider the finite alphabets $\mathcal{X} \triangleq \{0,1\}$,  $\mathcal{Y}$, and  $(\mathcal{Z}_s)_{s\in\mathfrak{S}}$, where $\mathfrak{S}$ is a finite set. Consider also the conditional probabilities $(p_{YZ(s)|X})_{s\in\mathfrak{S}}$. A wiretap channel is defined as a discrete memoryless channel with transition probability for one channel use $p_{YZ(s)|X}(y,z(s)|x)$ where
$x \in \mathcal{X}$ is the channel input from the transmitter, $y \in \mathcal{Y}$ is the channel output observed by the legitimate receiver, $z(s) \in \mathcal{Z}_s$ is the channel output observed by the eavesdropper, $s \in \mathfrak{S}$ is arbitrary, unknown to the legitimate users, and  can potentially change for each channel use. In the following, we omit the index $s\in \mathfrak{S}$ whenever $|\mathfrak{S}|=1$. 
Moreover, when the codeword $X^{1:N}$ is sent over the channel,  in addition to the channel output $Z^{1:N}(\mathbf{s})$, $\mathbf{s} \in \mathfrak{S}^N$, the eavesdropper has access to $X^{1:N}[\mathcal{S}]\triangleq (X^i)_{i \in \mathcal{S}}$, where $\mathcal{S} \subseteq \llbracket 1 , N \rrbracket$ is chosen by the eavesdropper and such that  $|\mathcal{S}| \triangleq \alpha N$ for some $\alpha \in [0,1]$. 

\begin{defn}
For $B \in \mathbb{N}$, define $\mathcal{B} \triangleq \llbracket 1 , B \rrbracket$. For $b \in \mathcal{B}$ and $R_b\geq 0$, define $R \triangleq \sum_{b \in \mathcal{B}}R_b/B$. A $(2^{NR},N, B)$ code has a rate $R$, operates over $B$ encoding blocks, and  consists for each encoding Block $b \in \mathcal{B}$ of
\begin{itemize}
\item A message set $\mathcal{M}_b \triangleq \llbracket 1 ,2^{NR_b} \rrbracket$.
\item A stochastic encoding function $f_b : \mathcal{M}_b \to \mathcal{X}^{N}$, used by the transmitter to encode a message $M_b$, uniformly distributed over $\mathcal{M}_b$, into $X_b^{1:N} \triangleq f_b(M_b)$. The messages $M_{1:B} \triangleq \left(M_b\right)_{b \in \mathcal{B}}$ are assumed mutually independent.
\item A deterministic decoding function used by the legitimate receiver $g_{b} : \mathcal{Y}^{N} \to \mathcal{M}_b$, to form $\widehat{M}_b$, an estimate of $M_b$, given the channel outputs $Y_b^{1:N}$. We write $\widehat{M}_{1:B}  \triangleq (\widehat{M}_{b})_{b \in \mathcal{B}}$.
\end{itemize}
\end{defn}

\begin{defn}
 A rate $R$ is achievable if there exists a sequence of $(2^{NR},N,B)$ codes such that 
\begin{align*}
\mathbb{P}[\widehat{M}_{1:B} \neq M_{1:B}] &\xrightarrow{{N \to \infty}}0,  \\
 \max_{ \mathbf{s} \in \mathfrak{S}^{NB},\mathcal{A} \in \mathbb{A}}\!\! I\left(M_{1:B};Z_{1:B}^{1:N}(\mathbf{s}), X_{1:B}^{1:N}[\mathcal{A}]\right) &\xrightarrow{{N \to \infty}}0, 
\end{align*}
where $\mathbb{A} \triangleq \! \{ (\mathcal{A}_b)_{b \in \mathcal{B}} : \mathcal{A}_b \subseteq \llbracket 1 , N \rrbracket \text{ and }|\mathcal{A}_b| = \alpha N, \forall b \in \mathcal{B} \}$, $(Z_{b}^{1:N}(\mathbf{s}_b),X_{b}^{1:N}[\mathcal{A}_b])$ corresponds to the random variables in Block $b \in \mathcal{B}$ for $\mathcal{A} = (\mathcal{A}_b)_{b \in \mathcal{B}} \in \mathbb{A}$ and $\mathbf{s}_b \in \mathfrak{S}^{N}$, $X_{1:B}^{1:N}[\mathcal{A}] \triangleq (X_{b}^{1:N}[\mathcal{A}_b])_{b \in \mathcal{B}}$, and $Z_{1:B}^{1:N}(\mathbf{s}) \triangleq (Z_{b}^{1:N}(\mathbf{s}_b))_{b \in \mathcal{B}}$ for $\mathbf{s} = (\mathbf{s}_b)_{b \in \mathcal{B}} \in \mathfrak{S}^{NB}$.\\
The supremum of such achievable rates is called secrecy capacity and denoted by $C_{s}$.
\end{defn}
When $\alpha =0$ and $|\mathfrak{S}| = 1$, our model recovers Wyner's wiretap channel \cite{Wyner75}. 
When $\alpha =0$ and $\mathbf{s} = (\mathbf{s}_b)_{b\in \mathcal{B}} \in \mathfrak{S}^{NB}$ is unknown to the legitimate users but all the components of $\mathbf{s}_b$, $b \in \mathcal{B}$, are identical, our model recovers a wiretap channel with a compound model for the eavesdropper's channel \cite{Liang09,bjelakovic2013secrecy}; the general model, as introduced in \cite{Liang09}, with compound models for both the eavesdropper's channel and the main channel is treated in Section \ref{secext}. When $\alpha =0$ and $\mathbf{s} = (\mathbf{s}_b)_{b\in \mathcal{B}} \in \mathfrak{S}^{NB}$ is unknown to the legitimate users, our model recovers a wiretap channel with an arbitrarily varying eavesdropper's channel \cite{molavianjazi2009arbitrary}.
When $\alpha >0$ and $|\mathfrak{S}| = 1$, our model recovers a special case of the wiretap channel of type~II~\cite{ozarow1984wire} when  $p_{Z|X} = p_{Z}$ and  $p_{Y|X}(y|x) = \mathds{1} \{ y=x\}, \forall (x,y)\in\mathcal{X} \times \mathcal{Y}$, a special case of the  wiretap channel of type II with noisy main channel~\cite{goldfeld2016semantic} when $p_{Z|X} = p_{Z}$, and a special case of the hybrid Wyner's/type II wiretap channel~\cite{nafea2016new}. Specifically, the difference between our model and the models in \cite{ozarow1984wire,nafea2016new,goldfeld2016semantic} is that, in our model, the eavesdropper observes a fraction $\alpha$ of each codeword $X_b^{1:N}$, $b\in\mathcal{B}$, whereas in \cite{ozarow1984wire,nafea2016new,goldfeld2016semantic}, the eavesdropper would be able to observe a fraction $\alpha$ of all the codewords considered jointly, i.e., $(X_b^{1:N})_{b\in\mathcal{B}}$. While the original type II constraint of~\cite{ozarow1984wire} is stronger than a blockwise type~II constraint, the latter constraint is relevant to model side-channel attacks where the eavesdropper is able to learn a bounded fraction of each codeword sent over the channel.

We now review the best known achievable rates for special cases of our model.

\begin{thm}[\!\!\cite{Wyner75,Csiszar78}] \label{th:WynerWTC}
Suppose that $|\mathfrak{S}| = 1$, and $\alpha =0$. Then, the secrecy capacity is
\begin{align*}
C_s = \max_{ \substack{ U - X -(Y,Z) \\ |\mathcal{U}| \leq |\mathcal{X}| }}  \left[ I(U;Y) - I(U;Z) \right]^+. 
\end{align*}
\end{thm}

\begin{thm}[\!\!\cite{ozarow1984wire}] \label{th:WTII}
Suppose that $|\mathfrak{S}| = 1$, $p_{Z|X} = p_{Z}$, and for any $x\in\mathcal{X}$, $y\in \mathcal{Y}$, $p_{Y|X}(y|x) = \mathds{1} \{ y=x\}$. Then,  
\begin{align*}
C_s = 1-\alpha.
\end{align*}
\end{thm}

\begin{thm}[\!\!{\cite{goldfeld2016semantic}}]   \label{th:WTIIn}
Suppose that $|\mathfrak{S}| = 1$, and $p_{Z|X} = p_{Z}$. Then, 
\begin{align*}
C_s = \max_{ \substack{ U - X -Y \\ |\mathcal{U}| \leq |\mathcal{X}| }} \left[ I(U;Y) - \alpha I(U;X) \right]^+. 
\end{align*}
\end{thm}

\begin{thm}[\!\!\cite{nafea2016new}]  \label{th:WTIIh}
Suppose that $|\mathfrak{S}| = 1$. Then,
\begin{align*}
C_s = \! \max_{ \substack{ U - X -(Y,Z) \\ |\mathcal{U}| \leq |\mathcal{X}| }}  \left[ I(U;Y) - \alpha I(U;X) - (1-\alpha) I(U;Z)\right]^+\!. 
\end{align*}
\end{thm}

\begin{thm}[\!\!\cite{Liang09,bjelakovic2013secrecy}]  \label{th:cap}
Consider the wiretap channel with compound eavesdropper channel statistics, i.e., assume that  $\mathbf{s} = (\mathbf{s}_b)_{b\in \mathcal{B}} \in \mathfrak{S}^{NB}$ is unknown to the legitimate users but all the components of $\mathbf{s}_b$, $b \in \mathcal{B}$, are identical. Assume also that $\alpha = 0$. Then,
\begin{align*}
C_s \geq \max_{ \substack{\forall s \in \mathfrak{S},U - X -(Y,Z(s)) \\ |\mathcal{U}| \leq |\mathcal{X}| }}   \min_{s \in \mathfrak{S}} \left[ I(U;Y) - I(U;Z(s)) \right]^+. 
\end{align*}
Moreover, for a degraded channel, i.e., when for all $s \in \mathfrak{S}$, $X-Y-Z(s)$, we have $$
C_s = \max_{ p_X}  \min_{s \in \mathfrak{S}} I(X;Y |Z(s)) .  $$
\end{thm}

\begin{thm}[\!\!\cite{molavianjazi2009arbitrary,bjelakovic2013capacity}] \label{th:arb}
Consider the wiretap channel with arbitrarily varying eavesdropper channel, i.e., assume that $\mathbf{s} \in \mathfrak{S}^{NB}$ is unknown to the legitimate users. Assume also that $\alpha = 0$. Define $\overline{\mathfrak{S}}$ as the set of all the convex combinations of elements of ${\mathfrak{S}}$. If there exists a best channel for the eavesdropper, i.e., $ \exists s^*\in \overline{\mathfrak{S}}, \forall s \in \mathfrak{S}$, $X-Z (s^*) -Z(s)$,~then 
\begin{align*}
C_s \geq  \max_{ \substack{\forall \bar{s} \in \overline{\mathfrak{S}}, U - X -(Y,Z(\bar{s})) \\ |\mathcal{U}| \leq |\mathcal{X}| }} \min_{\bar{s} \in \overline{\mathfrak{S}}} \left[ I(U;Y) - I(U;Z(\bar{s})) \right]^+. 
\end{align*}
Moreover, if there exists a best channel for the eavesdropper and for all $\bar{s} \in \overline{\mathfrak{S}}$, $X-Y-Z(\bar{s})$, then
$$
C_s =  \max_{ p_X} \min_{\bar{s} \in \overline{\mathfrak{S}}} I(X;Y|Z(\bar{s}) ). 
$$
\end{thm}

 Note that \cite{nafea2016new,Liang09,bjelakovic2013secrecy,molavianjazi2009arbitrary,bjelakovic2013capacity} prove
the existence of coding schemes that achieve the rates in Theorems \ref{th:WTIIn}-\ref{th:arb} but do not provide explicit
coding schemes. To the best of our knowledge, no explicit coding schemes that achieve the secrecy rates in  Theorems~\ref{th:WTIIn}-\ref{th:arb} have been previously proposed.

More specifically,~\cite{renes2013efficient,Gulcu14,Chou16} provided polar coding schemes that achieve the strong secrecy capacity for Wyner's wiretap channel, i.e., Theorem \ref{th:WynerWTC}, with the following caveats: a pre-shared secret with negligible rate is required in  \cite{Gulcu14,Chou16}, no efficient method is known to construct the codebooks in~\cite{renes2013efficient}, and  the existence of certain deterministic maps  is needed in \cite{Gulcu14}, similar to~\cite[Theorem~3]{Honda13}. 
Note that a main tool in \cite{Gulcu14,Chou16} is block-Markov coding to support non-degraded and non-symmetric channels. 
Using techniques similar  to \cite{renes2013efficient,Gulcu14,Chou16}, including block-Markov coding, and ideas for compound channels without security constraints in \cite{hassani2014universal,csacsouglu2016universal}, it is unclear to us how to extend existing polar coding schemes   to the wiretap channel models of Theorems~\ref{th:WTII}-\ref{th:cap} and Theorems~\ref{thmain1}, \ref{thmain2}, \ref{thmainb} because of the uncertainty on the eavesdropper's observations.

A different approach than polar coding to obtain wiretap codes for Wyner's wiretap channel is provided in \cite{hayashi2010construction,Bellare2012,Hayashi11}. Specifically, these works construct wiretap codes using (i) capacity-achieving channel codes (without security constraint), and (ii) universal hashing~\cite{Bennett95}, and have been the first works to provide efficient codes that asymptotically achieve  optimal secrecy rates and strong secrecy for additive or symmetric and degraded wiretap channels. \cite{hayashi2016secure} subsequently extended these constructions to any wiretap~channels as in Theorem \ref{th:WynerWTC}.

It is also worth noting that \cite{renes2011noisy} proposed wiretap channel coding for Wyner's model using source coding with side information and universal hashing. It is, however, unclear to us how to directly translate the scheme of \cite{renes2011noisy} to an efficient code construction without employing block-Markov coding for the part of the coding scheme that involves source coding with side~information.

Our approach in this paper departs from the works in \cite{hayashi2010construction,Bellare2012,Hayashi11,hayashi2016secure}
 because, instead of relying on channel codes, we rely on source codes to handle the reliability constraint, which allows us to use a block-Markov coding approach to handle non-symmetric and non-degraded channels. Our approach also departs from existing polar coding schemes, as our construction solely relies on polar coding results for source coding with side information, does not require the existence of certain maps, and does not require a pre-shared key to ensure strong secrecy. In addition to proposing the first explicit coding schemes that achieve the secrecy rates in  Theorems~\ref{th:WTIIn}-\ref{th:arb} and Theorems \ref{thmain1}, \ref{thmain2}, \ref{thmainb}, our coding approach also proves that all the models  considered in this paper can be treated under a unified framework that only requires three primitives:  (i) source coding with side information, (ii)~universal hashing, and (iii) distribution approximation.

\section{Statement of Main Results} \label{sec:results}

Our main results are the following theorems. 
\begin{thm} \label{thmain1}
If all the components of $\mathbf{s}_b$, $b \in \mathcal{B}$, are identical, then the coding scheme of Section \ref{sec:preshared} achieves the secrecy rate
$$ \max_U\left[ I(U;Y) - \alpha I(U;X) - (1-\alpha) \max_{s\in \mathfrak{S}} I \left({U}; {Z}(s) \right) \right]^+\!\!,$$ where the maximum is taken over $U$ such that $ \forall s \in \mathfrak{S},U - X -(Y,Z(s))$ and $|\mathcal{U}| \leq |\mathcal{X}| $.
  \end{thm}
  
  \begin{thm} 
   Assume that the components of $\mathbf{s}_b$, $b \in \mathcal{B}$, are arbitrary. If  there exists a best channel for the eavesdropper, then the coding scheme of Section \ref{sec:preshared} achieves the secrecy rate
$$ \max_U \left[ I(U;Y) - \alpha I(U;X) - (1-\alpha) \max_{\bar{s}\in \bar{\mathfrak{S}}} I \left({U}; {Z}(\bar{s}) \right) \right]^+,\!\!$$ 
where the maximum is taken over $U$ such that $ \forall \bar{s} \in \overline{\mathfrak{S}}, U - X -(Y,Z(\bar{s}))$ and $|\mathcal{U}| \leq |\mathcal{X}| $.
\label{thmain2}
\end{thm}

The proof of Theorem \ref{thmain1} is presented in two parts. First, in Section \ref{sec:proof1a}, the initialization phase, i.e., Algorithms \ref{alg:p3}, \ref{alg:p4}, is ignored and Theorem \ref{thmain1} is proved under the assumption that the legitimate users have a pre-shared key whose rate is negligible. Next, in Section~\ref{sec:initphase}, Theorem \ref{thmain1} is proved without this assumption by considering the initialization phase combined with Algorithms \ref{alg:p1}, \ref{alg:p}. The proof of Theorem \ref{thmain2} is similar to the one of  Theorem \ref{thmain1} and is discussed in Section~\ref{sec:proof2}.

Finally, from Theorems \ref{thmain1} and \ref{thmain2}, we conclude that the secrecy rates of Theorems~\ref{th:WynerWTC}-\ref{th:arb}  are achieved. 

Note that we will also extend Theorem \ref{thmain1} to the case of a compound main channel in Theorem \ref{thmainb}.

\section{Coding scheme} \label{sec:preshared}

Our coding scheme consists of two phases: An initialization phase presented in Section \ref{sec:initk}, and the actual secure communication phase presented in Section \ref{secCS}. The initialization phase allows the legitimate users to share a secret key which is used in the second phase of the coding scheme. Both phases rely on three primitives presented in Section \ref{sec:notation}.

In this section, for $s \in \mathfrak{S}$,  we consider an arbitrary joint distribution $q_{UXYZ(s)} \triangleq q_{UX} p_{YZ(s)|X}$  with $|\mathcal{U}| = |\mathcal{X}| = 2 $ and such that $U - X -(Y,Z(s))$. Let $K$ be a power of two, let $(U^{1:K},X^{1:K})$ be distributed according to $q_{U^{1:K}X^{1:K}} \triangleq \prod_{i= 1}^K q_{UX}$, and define  $A^{1:K} \triangleq  U^{1:K} G_K$, $V^{1:K} \triangleq  X^{1:K} G_K$, where $G_K \triangleq \left[ \begin{smallmatrix}
       1 & 0            \\[0.3em]
       1 & 1 
     \end{smallmatrix} \right]^{\otimes \log K} $ is the matrix defined in~\cite{Arikan10}. Define also for $\delta_K \triangleq 2^{-K^{\beta}}$, $\beta \in ]0,1/2[$, the sets
\begin{align*}
\mathcal{V}_U &\triangleq \left\{ i \in \llbracket 1, K \rrbracket : H( A^i | A^{1:i-1}) >  1 - \delta_K \right\}, \\
\mathcal{H}_{U} &\triangleq \left\{ i \in \llbracket 1, K \rrbracket : H( A^i | A^{1:i-1} ) > \delta_K \right\}, \\
\mathcal{V}_{U|Y} &\triangleq \left\{ i \in \llbracket 1, K \rrbracket : H( A^i | A^{1:i-1} Y^{1:K}) > 1 - \delta_K \right\},\\
\mathcal{H}_{U|Y} &\triangleq \left\{ i \in \llbracket 1, K \rrbracket : H( A^i | A^{1:i-1} Y^{1:K}) > \delta_K \right\}, \\
     \mathcal{V}_{X} &\triangleq \left\{ i \in \llbracket 1, K \rrbracket : H( V^i | V^{1:i-1} ) >1- \delta_K \right\}, \\
\mathcal{V}_{X|U} &\triangleq \left\{ i \in \llbracket 1, K \rrbracket : H( V^i | V^{1:i-1} U^{1:K}) > 1-\delta_K \right\}. 
\end{align*}

\subsection{Primitives used in the coding scheme} \label{sec:notation}
 
\noindent{}\textbf{Primitive 1}: Source coding (\textbf{SC}) with side information for the source $(\mathcal{U} \times \mathcal{Y},q_{UY})$ \cite{Arikan10}. 
Define the encoder $f^{\textbf{SC}} \triangleq (f^{\textbf{SC}}_1,f^{\textbf{SC}}_2)$ with
\begin{align*}
f^{\textbf{SC}}_1(A^{1:K}) & \triangleq A^{1:K}[\mathcal{V}_{U|Y} ] ,\\
f^{\textbf{SC}}_2(A^{1:K}) & \triangleq A^{1:K}[\mathcal{H}_{U|Y} \backslash \mathcal{V}_{U|Y} ].
\end{align*}
Then, define $g^{\textbf{SC}}$ as the successive cancellation decoder of~\cite{Arikan10} such that if $\widehat{A}^{1:K} \triangleq g^{\textbf{SC}} (f^{\textbf{SC}}_1(A^{1:K}),f^{\textup{\textbf{SC}}}_2(A^{1:K}),Y^{1:K})$, then
\begin{align}
\mathbb{P} [ \widehat{A}^{1:K} \neq {A}^{1:K}] \leq K \delta_K. \label{lemscs}
\end{align}
\begin{rem}
We decompose $f^{\textbf{SC}}$ in two parts $f^{\textbf{SC}}_1$ and $f^{\textbf{SC}}_2$ because $f^{\textbf{SC}}_1(A^{1:K})$ can be shown to be almost uniform in divergence, e.g., \cite[Lemma 8]{chou2015coding}, which will be a useful property in our coding scheme analysis. Note, however, that the distribution of $
(f^{\textbf{SC}}_1(A^{1:K})\lVert f^{\textbf{SC}}_2(A^{1:K}))$ is not necessarily close to a uniform distribution.
 \end{rem}
\medskip
\noindent{}\textbf{Primitive 2}: Universal hashing (\textbf{UH}) \cite{Carter79}. Let $ c,d \in \mathbb{N}$ such that $ d \leq c$, 
and define $\mathcal{S}\triangleq\{0,1\}^{c} \backslash \{\mathbf{0}\}$. Then, define for $S\in \mathcal{S}$, $T \in \{0,1\}^{c}$, $R \in \{0,1\}^{d}$,  $R' \in \{0,1\}^{c - d}$
\begin{align*}
f^{\textbf{UH}}_S( R, R') & \triangleq S^{-1} \odot (R \lVert R') ,\\
g^{\textbf{UH}}_S(T,d) & \triangleq (S\odot T)_{d}, 
\end{align*}
where $\odot$ is the multiplication in $\textup{GF}(2^{c})$ and $(\cdot)_{d}$ selects the $d$ most significant bits, such that
$$
g^{\textbf{UH}}_S(f^{\textbf{UH}}_S( R, R'), d) = R.
$$
By \cite{Bellare2012}, $\mathcal{F}\triangleq\{g^{\textbf{UH}}_S\}_{S\in \mathcal{S}}$ is a family of two-universal hash~functions. 

\medskip
\noindent{}\textbf{Primitive 3}:
 Distribution approximation (\textbf{DA}) for  $q_{A^{1:K}}$, the distribution of $A^{1:K} \triangleq  U^{1:K} G_K$, where $U^{1:K}$ follows $q_{U^{1:K}} \triangleq \prod_{i=1}^K q_U$. Let  $T^{1:|\mathcal{V}_U|}$ be a sequence of uniformly distributed bits over $\{0,1\}^{|\mathcal{V}_U|}$. Then, define $\widetilde{A}^{1:K}$ according to the distribution $\widetilde{p}_{A^{1:K}} \triangleq \prod_{j=1}^K \widetilde{p}_{{A}^j|{A}^{1:j-1}} $ with 
         \begin{align} \label{eqeqa}
\widetilde{p}_{{A}^j|{A}^{1:j-1}} ({a}^j|{a}^{1:j-1}) \triangleq    \begin{cases} 
\mathds{1} \{ a^j = T^j\} &\text{if }j\in {\mathcal{V}}_{U}  \\
{q}_{A^j|A^{1:j-1}} (a^j|a^{1:j-1}) &\text{if }j\in {\mathcal{V}}_{U}^c \end{cases}  
  \end{align}
  We write $\widetilde{A}^{1:K} = f^{\textbf{DA}}(T^{1:|\mathcal{V}_U|})$. Moreover, we have 
 \begin{align}
 \mathbb{D}(q_{A^{1:K}}\lVert \widetilde{p}_{A^{1:K}} ) \nonumber 
& \stackrel{(a)}{=} \sum_{j=1}^K \mathbb{E}_{q_{A^{1:j-1}}} \mathbb{D}( q_{A^{j}|A^{1:j-1}} \lVert \widetilde{p}_{A^{j} | A^{1:j-1}} ) \nonumber  \\
& \stackrel{(b)}{=}  \sum_{j\in \mathcal{V}_U} (1 - H(A^j |A^{1:j-1}) )  \stackrel{(c)} \leq K \delta_K,  \label{eqDA1}
\end{align}  
where $(a)$ holds by the chain rule, $(b)$ holds by \eqref{eqeqa}, $(c)$ holds by the definition of $\mathcal{V}_U$.
  
\medskip
\noindent{}\textbf{Variant of Primitive 3}: Channel prefixing (\textbf{CP}) for the distribution $q_{X^{1:K}U^{1:K}} \triangleq \prod_{i=1}^K q_{XU}$. Given $U^{1:K}$ distributed according to $q_{U^{1:K}}$, define $\widetilde{V}^{1:K}$ according to the distribution $\widetilde{p}_{U^{1:K}V^{1:K}} \triangleq q_{U^{1:K}} \prod_{j=1}^K \widetilde{p}_{{V}^j|{V}^{1:j-1}U^{1:K}} $ with  
 \begin{align}
 &\widetilde{p}_{V^j|V^{1:j-1}U^{1:K}} (v^j|v^{1:j-1}\widetilde{u}^{1:K}) \nonumber\\
 & \triangleq
\begin{cases}
  1/2 & \text{if }j \in  \mathcal{V}_{X|U}\\
  {q}_{V^j|V^{1:j-1}U^{1:K}} (v^j|v^{1:j-1}\widetilde{u}^{1:K}) & \text{if }j\in \mathcal{V}_{X|U}^c
 \end{cases}  \label{eqeqvu}  
\end{align}
 We write $\widetilde{V}^{1:K} = f^{\textbf{CP}}(U^{1:K})$. Moreover, we have 
\begin{align}
& \mathbb{D}( q_{U^{1:K}V^{1:K}} \lVert \widetilde{p}_{U^{1:K}V^{1:K}} ) \displaybreak[0] \nonumber\\
& \stackrel{(a)}{=} \sum_{j=1}^K  \mathbb{E}_{q_{U^{1:K}V^{1:j-1}}} \mathbb{D}( q_{V^{j}|V^{1:j-1}U^{1:K}} \lVert \widetilde{p}_{V^{j}|V^{1:j-1}U^{1:K}} )  \nonumber  \displaybreak[0] \\
& \stackrel{(b)}{=}   \sum_{j\in \mathcal{V}_{X|U}}  (1- H(V^{j}|V^{1:j-1}U^{1:K} ) )  \stackrel{(c)} \leq  K \delta_K,  \label{eqCP}
\end{align} 
where $(a)$ holds by the chain rule, $(b)$ holds by \eqref{eqeqvu}, $(c)$ holds by the definition of $\mathcal{V}_{X|U}$.

\subsection{Coding scheme: Phase I - Initialization} \label{sec:initk}

The legitimate users create a secret key with length $l_{\textup{key}}$, 
which will be be specified later in Section \ref{sec:keysec}, with Algorithms \ref{alg:p3} and \ref{alg:p4}, which operate over $B_0$ blocks of length $N \triangleq KL$, where $L,K \in \mathbb{N}$, and $K$ is a power of two.  We define $\mathcal{B}_0 \triangleq \llbracket 1 , B_0 \rrbracket$ and $\mathcal{L} \triangleq \llbracket 1 , L \rrbracket$. 
In each Block~$b \in \mathcal{B}_0$, the encoder forms the key Key$_b$ with length $l_{\textup{key}}' \triangleq l_{\textup{key}}/B_0$, as described in Algorithm~\ref{alg:p3}.  The encoder uses  the following randomization sequences: $R^{\textup{init}'}_b \triangleq (R^{\textup{init}'}_{b,l})_{l\in\mathcal{L}}$, where $R^{\textup{init}'}_{b,l}$, $l\in\mathcal{L}$, is a sequence of uniformly distributed bits over $\{0,1 \}^{|\mathcal{H}_{U|Y} |-| \mathcal{V}_{U|Y}|} $,  $R^{\textup{init}}_b$ is a sequence of uniformly distributed bits over $ \mathcal{R}^{\textup{init}}\triangleq \{ 0,1\}^{N}  \backslash \{ \mathbf{0} \}$. The encoder also uses the local randomness $(R^{\textup{loc}}_{b,l})_{l \in \mathcal{L}}$, where $R^{\textup{loc}}_{b,l}$, $l \in \mathcal{L}$, is a sequence of uniformly distributed bits over $\{0,1\}^{|\mathcal{V}_U|}$.

\begin{algorithm}
  \caption{Initialization at the transmitter}
  \label{alg:p3}
  \begin{algorithmic}[1] 
\REQUIRE Randomization sequences $(R^{\textup{init}}_b)_{b \in \mathcal{B}_0}$ and $(R^{\textup{init}'}_b)_{b \in \mathcal{B}_0}$ \vspace{0.1em}
\FOR{Block $b \in \mathcal{B}_0$}  
    \FOR{Sub-block $l \in \mathcal{L}$}
        
 \STATE Define $\widetilde{A}_{b,l}^{1:K} \triangleq f^{\textbf{DA}}(R^{\textup{loc}}_{b,l})$  
  	\STATE Define $\widetilde{U}_{b,l}^{1:K} \triangleq \widetilde{A}_{b,l}^{1:K} G_K$ 
  	  \STATE Define $\widetilde{V}_{b,l}^{1:K}\triangleq f^{\textbf{CP}} ( \widetilde{U}_{b,l}^{1:K} )$ 
  	\STATE Define $\widetilde{X}_{b,l}^{1:K} \triangleq \widetilde{V}_{b,l}^{1:K} G_K$
	\ENDFOR
	\STATE Transmit $\widetilde{X}^{1:N}_{b}  \triangleq \displaystyle\concat_{l\in \mathcal{L}} \widetilde{X}_{b,l}^{1:K}$ over the channel
	\STATE Let $\widetilde{Y}^{1:N}_{b} \triangleq \displaystyle\concat_{l\in \mathcal{L}} \widetilde{Y}_{b,l}^{1:K}$, $\widetilde{Z}^{1:N}_{b}(\mathbf{s}_b) \triangleq \displaystyle\concat_{l\in \mathcal{L}} \widetilde{Z}_{b,l}^{1:K}(\mathbf{s}_{b,l})$ denote the channel outputs
		\STATE Transmit with a channel code \cite{Arikan09}\\ $D_b \triangleq \displaystyle\concat_{l\in \mathcal{L}}  \left[ \left(f^{\textbf{SC}}_2(\widetilde{A}^{1:K}_{b,l}) \oplus R^{\textup{init}'}_{b,l} \right) \concat f^{\textbf{SC}}_1(\widetilde{A}^{1:K}_{b,l}) \right],$\\
		where $\oplus$ denotes modulo 2 addition 
		\STATE Define $\widetilde{U}^{1:N}_{b} \triangleq \displaystyle\concat_{l\in \mathcal{L}}\widetilde{U}_{b,l}^{1:K}$ 
		\STATE Define $ \textup{Key}_b \triangleq g^{\textbf{UH}}_{R^{\textup{init}}_b} ( \widetilde{U}^{1:N}_{b}, l_{\textup{key}}')$ 
	\ENDFOR
  \end{algorithmic}  
\end{algorithm}
\begin{algorithm}
  \caption{Initialization phase at the receiver}
  \label{alg:p4}
  \begin{algorithmic}[1] 
\REQUIRE  $(R^{\textup{init}}_b)_{b \in \mathcal{B}_0}$ and $(R^{\textup{init}'}_b)_{b \in \mathcal{B}_0}$ 
\FOR{Block $b \in \mathcal{B}_0$}  
\STATE Form an estimate $\widehat{D}_b$ of $D_b$  
    \FOR{Sub-block $l \in \mathcal{L}$}
\STATE Given $(\widehat{D}_b,R^{\textup{init}'}_b)$ and Line 10 of Algorithm \ref{alg:p3}, form an estimate of $(f^{\textbf{SC}}_1(\widetilde{A}^{1:K}_{b,l}), f^{\textbf{SC}}_2(\widetilde{A}^{1:K}_{b,l}))$ and denote this estimate by 
$(\widehat{A}^{1:K}_{b,l}[\mathcal{V}_{U|Y}],\widehat{A}^{1:K}_{b,l}[\mathcal{H}_{U|Y} \backslash \mathcal{V}_{U|Y}])$ 
         \STATE 
         Form an estimate of $\widetilde{A}^{1:K}_{b,l}$ as $$\widehat{A}^{1:K}_{b,l} \triangleq g^{\textbf{SC}}(\widehat{A}^{1:K}_{b,l}[\mathcal{V}_{U|Y}],\widehat{A}^{1:K}_{b,l}[\mathcal{H}_{U|Y} \backslash \mathcal{V}_{U|Y}], \widetilde{Y}^{1:K}_{b,l})$$
         \vspace{-1em}       
         \STATE Form $\widehat{U}_{b,l}^{1:K} \triangleq \widehat{A}_{b,l}^{1:K} G_K$ an estimate of $\widetilde{U}_{b,l}^{1:K} $
	\ENDFOR
\STATE	Form $\widehat{U}^{1:N}_{b}  \triangleq \displaystyle\concat_{l\in \mathcal{L}} \widehat{U}_{b,l}^{1:K}$ an estimate of $\widetilde{U}^{1:N}_{b} $ 
\STATE	Form $ \widehat{\textup{Key}}_b = g^{\textbf{UH}}_{R^{\textup{init}}_b} ( \widehat{U}^{1:N}_{b} , l_{\textup{key}}')$ an estimate of $ { \textup{Key}_b} $
	\ENDFOR
  \end{algorithmic}  
\end{algorithm}
\begin{rem} \label{rem1}
In Line 10 of Algorithm \ref{alg:p3}, note that the channel code \cite{Arikan09} requires a uniformly distributed message. While $\concat_{l\in \mathcal{L}}  \widetilde{A}^{1:K}_{b,l}[\mathcal{H}_{U|Y}]$ is not a sequence of uniformly distributed bits, $D_b$ is a sequence of uniformly distributed bits over $\llbracket 1, 2^{L|\mathcal{H}_{U|Y} |} \rrbracket$. 
\end{rem}

\textbf{High-level description of the initialization phase}: The initialization phase  is depicted in Figure~\ref{figinit} and consists in $B_0$ communication blocks. All the communication blocks are independent, and each Block $b\in\mathcal{B}_0$ will lead to the exchange of a key Key$_b$ between the legitimate users, which will be shown to be secret from the eavesdropper. Additionally, $B_0$ is chosen such that the length of the keys $(\textup{Key}_b)_{b \in \mathcal{B}_0}$ is sufficiently large to be used in  the main coding scheme, which is described in the next section and allows the exchange of a secret message between the legitimate users. It will also be shown that the initialization phase considered jointly with the main coding scheme has a negligible effect on the overall communication rate and the overall information leakage to the~eavesdropper.
\begin{figure} 
\centering
 \includegraphics[width=8 cm]{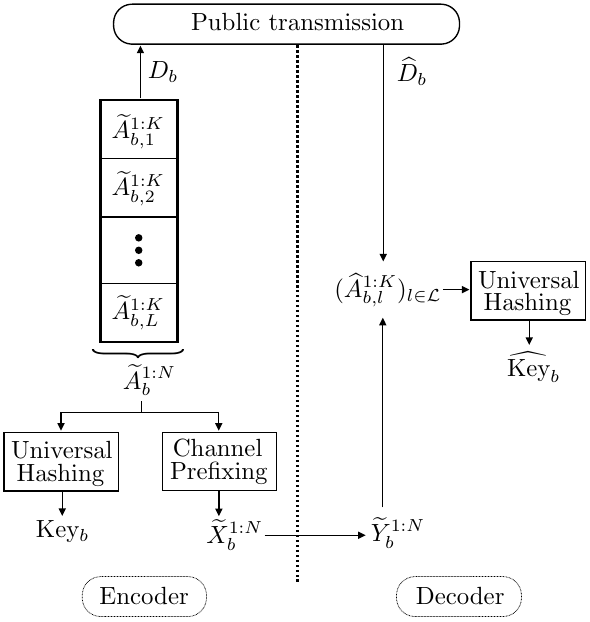}
  \caption{Initialization phase for Block $b \in \mathcal{B}_0$. The encoder creates $\widetilde{A}^{1:N}_b$, which is made of $L$ sub-blocks $(\widetilde{A}^{1:K}_{b,l})_{l \in \mathcal{L}}$. Then, from $\widetilde{A}^{1:N}_b$, the encoder creates Key$_b$ (by universal hashing), and the codeword $\widetilde{X}^{1:N}_b$  (via channel prefixing), which is sent over the channel and whose noisy observation by the legitimate receiver is $\widetilde{Y}^{1:N}_b$. The decoder creates an estimate of $\widetilde{A}^{1:N}_b$ from $\widetilde{Y}^{1:N}_b$ and an estimate of ${D}_b$, which is sent to him via a channel code, as described in Line 10 of Algorithm \ref{alg:p3}. Finally,  the decoder creates $\widehat{\textup{Key}}_b$, an estimate of Key$_b$, from his estimate of $\widetilde{A}^{1:N}_b$.}
  \label{figinit}
\end{figure}

 Consider Block $b\in\mathcal{B}_0$ in Algorithm \ref{alg:p3}. As described in Lines 3-4, the encoder creates $(\widetilde{U}_{b,l}^{1:K})_{l \in \mathcal{L}}$  such that the distribution of $(\widetilde{U}_{b,l}^{1:K})_{l \in \mathcal{L}}$ is close to the product distribution $q_{U^{1:N}}$. 
Then, as described in Lines 5-6, channel prefixing is performed to create from $(\widetilde{U}_{b,l}^{1:K})_{l \in \mathcal{L}}$ the codewords $( \widetilde{X}_{b,l}^{1:K})_{l \in \mathcal{L}}$ that are sent over the channel, and whose noisy observations at the legitimate receiver are  $(\widetilde{Y}_{b,l}^{1:K})_{l \in \mathcal{L}}$. Additionally, the key $\textup{Key}_b$ is formed from $(\widetilde{U}_{b,l}^{1:K})_{l \in \mathcal{L}}$ through universal hashing, as described in Line 12. As shown later, secrecy of the key is ensured via an appropriate choice of the hash function output length. As described in Line 10, the encoder sends $D_b$ to the legitimate receiver using a regular channel code (without security guarantees) - see also Remark~\ref{rem1}.

 Finally, as described in Lines~2-7 of Algorithm \ref{alg:p4}, upon estimating  $D_b$, the legitimate receiver  forms an estimate of $(\widetilde{U}_{b,l}^{1:K})_{l \in \mathcal{L}}$ from $(\widetilde{Y}_{b,l}^{1:K})_{l \in \mathcal{L}}$.   Then, as described in Line 9 of Algorithm \ref{alg:p4}, from the estimate of $(\widetilde{U}_{b,l}^{1:K})_{l \in \mathcal{L}}$, the legitimate receiver creates an estimate of Key$_b$.

\subsection{Coding scheme: Phase II - Secure communication} \label{secCS}
The encoding scheme operates over $B$ blocks of length $N \triangleq KL$, where $L,K \in \mathbb{N}$ and $K$ is a power of two.  We define $\mathcal{B} \triangleq \llbracket 1 , B \rrbracket$ and $\mathcal{L} \triangleq \llbracket 1 , L \rrbracket$. Encoding at the transmitter and decoding at the receiver are described in Algorithms \ref{alg:p1} and~\ref{alg:p}, respectively. 
In each block $b \in \mathcal{B}$, the transmitter encodes, as described in Algorithm~\ref{alg:p1}, a message $M_b$ uniformly distributed over $\llbracket 1 , 2^{|M_b|} \rrbracket$ and represented by a binary sequence with length 
\begin{equation*} 
|M_b| \triangleq  \begin{cases} |M_1| \text{ if } b=1\\ |M_1| - L |\mathcal{V}_{U|Y}| \text{ otherwise } \end{cases}.
\end{equation*}
Algorithms \ref{alg:p1} and \ref{alg:p} depend on the parameter 
\begin{equation} \label{eqparamr}
r \triangleq   |M_1|,
\end{equation}
 which will be specified later. 
 
In each block $b \in \mathcal{B}$, as described in Algorithm~\ref{alg:p1}, the encoder uses the local randomness $R_b'$, a binary randomization sequence uniformly distributed over $\llbracket 1 , 2^{|R_b'|} \rrbracket$. The sequences  $R_{1:B}' \triangleq \left(R_b'\right)_{b \in \mathcal{B}}$ are mutually independent.
The length of the sequences $\left(R_b'\right)_{b \in \mathcal{B}}$ is defined for $b \in \mathcal{B}$ as $|R_b'| \triangleq L |\mathcal{V}_{U}| -r.$
In each block $b \in \mathcal{B}$, the encoder also uses, as described in Algorithm~\ref{alg:p1}, $R_b$, a binary randomization sequence with length $L|\mathcal{V}_U|$, uniformly distributed over $\mathcal{R} \triangleq \{0,1\}^{L|\mathcal{V}_U|} \backslash \{\mathbf{0}\}$. The sequences  $R_{1:B} \triangleq \left(R_b\right)_{b \in \mathcal{B}}$ are mutually independent. Moreover, it is assumed that $M_{1:B}$, $R_{1:B}$, and $R_{1:B}'$ are mutually independent.

\begin{rem} \label{rem1a}
In Algorithm \ref{alg:p1}, observe that $T^{1:|\mathcal{V}_U|L}_b$, $b \in \mathcal{B}$, is uniformly distributed over $ \{0,1\}^{|\mathcal{V}_U|L}$ because $(M_b\lVert M'_b\lVert R_b')$ is uniformly distributed over $\{0,1\}^{|\mathcal{V}_U|L}$ and independent of $R_b$. Hence, the $L$ random variables $(T^{1:|\mathcal{V}_U|}_{b,l})_{l \in \mathcal{L}}$ are uniformly distributed over $\{0,1\}^{|\mathcal{V}_U|}$ and independent. When the elements of $\mathbf{s}_b$ are all equal to $s$, then, by construction, the conditional probability $\widetilde{p}_{{Z}^{1:K}_{b,l}(s)|T^{1:|\mathcal{V}_U|}_{b,l}}$ is the same for all $l \in \mathcal{L}$, and the $L$ pairs  $( (T^{1:|\mathcal{V}_U|}_{b,l}, \widetilde{Z}^{1:K}_{b,l}(s)) )_{l \in \mathcal{L}}$ are independently and identically distributed according to the joint distribution $\widetilde{p}_{{T}^{1:|\mathcal{V}_U|}_{b,1} {Z}^{1:K}_{b,1}(s)}$.	
\end{rem}
\begin{rem} \label{rem2}
In Algorithm \ref{alg:p1}, consider $\widetilde{X}^{1:K}_{b,l}[\mathcal{A}_{b,l}]$, $b \in \mathcal{B}$, $l\in\mathcal{L}$, where for all $l\in\mathcal{L}$, $\mathcal{A}_{b,l} \subset \llbracket 1,K \rrbracket$ and $\sum_{l \in \mathcal{L}} |\mathcal{A}_{b,l}| = \alpha N$ such that $\widetilde{X}^{1:N}_{b}[\mathcal{A}_b] \triangleq \lVert_{l \in \mathcal{L}} \widetilde{X}^{1:K}_{b,l}[\mathcal{A}_{b,l}]$ corresponds to the $\alpha N$ symbols of the codewords emitted at the transmitter that the eavesdropper has chosen to have access to. Similar to Remark~\ref{rem1a}, the $L$ triplets $( (T^{1:|\mathcal{V}_U|}_{b,l}, \widetilde{X}^{1:K}_{b,l}[\mathcal{A}_{b,l}], \widetilde{Z}^{1:K}_{b,l}(\mathbf{s}_{b,l})) )_{l \in \mathcal{L}}$ are independent, however, they are not necessarily identically distributed because the components of $\mathbf{s}_{b,l}$ are arbitrary, and because the sets $(\mathcal{A}_{b,l})_{l \in \mathcal{L}}$ are arbitrarily chosen by the eavesdropper. 
\end{rem}

\begin{algorithm}
  \caption{Encoding}
  \label{alg:p1}
  \begin{algorithmic}[1]
\REQUIRE Randomization sequences $(R_b)_{b \in \mathcal{B}}$, $(R_b')_{b \in \mathcal{B}}$, and messages $(M_b)_{b \in \mathcal{B}}$ \vspace{0.1em}
              \STATE Define $M'_1 \triangleq \emptyset$ \vspace{0.15em}
              \FOR{Block $b \in \mathcal{B}$}\vspace{0.1em}
    \STATE Define $M'_b \triangleq \displaystyle\concat_{l\in \mathcal{L}}f^{\textbf{SC}}_1 \left( \widetilde{A}^{1:K}_{b-1,l} \right)$ if $b\neq1$   \vspace{0.1em} 
    \STATE      Define    $T^{1:|\mathcal{V}_U|L}_b \triangleq f^{\textbf{UH}}_{R_b}( M_b, M'_b \lVert R_b')$ \vspace{0.15em}
    \FOR{Sub-block $l \in \mathcal{L}$} \vspace{0.1em}
         \STATE Consider  the notation  $T^{1:|\mathcal{V}_U|}_{b,l} \triangleq T^{(l-1)|\mathcal{V}_U| + 1: l |\mathcal{V}_U| }_b$ \vspace{0.1em}
        \STATE Define  $\widetilde{A}_{b,l}^{1:K} \triangleq  f^{\textbf{DA}}\left( T^{1:|\mathcal{V}_U|}_{b,l}\right)$  \vspace{0.1em}
         \STATE Define $\widetilde{U}_{b,l}^{1:K} \triangleq \widetilde{A}_{b,l}^{1:K} G_K $ \vspace{0.1em}
  \STATE Define $\widetilde{V}_{b,l}^{1:K} \triangleq f^{\textbf{CP}} \left( \widetilde{U}_{b,l}^{1:K} \right)$ \vspace{0.1em}
  	\STATE Define $\widetilde{X}_{b,l}^{1:K} \triangleq \widetilde{V}_{b,l}^{1:K} G_K$ \vspace{0.1em}
	\ENDFOR
	\STATE Transmit $\widetilde{X}^{1:N}_{b}  \!\! \triangleq \!\! \displaystyle\concat_{l\in \mathcal{L}} (\widetilde{X}_{b,l}^{1:K})$ over the channel \vspace{0.1em}
	\STATE Let $\widetilde{Y}^{1:N}_{b} \triangleq \displaystyle\concat_{l\in \mathcal{L}} \widetilde{Y}_{b,l}^{1:K}$, $\widetilde{Z}^{1:N}_{b}(\mathbf{s}_b) \triangleq \displaystyle\concat_{l\in \mathcal{L}} \widetilde{Z}_{b,l}^{1:K}(\mathbf{s}_{b,l})$ denote the channel outputs
	\ENDFOR	
	\STATE Using a pre-shared secret, apply a one-time pad to $(f^{\textbf{SC}}_2(\widetilde{A}^{1:K}_{b,l}))_{l\in \mathcal{L}, b\in \mathcal{B}}$, and $(f^{\textbf{SC}}_1(\widetilde{A}^{1:K}_{B,l}))_{l\in \mathcal{L}}$, then transmit the result with a channel code \cite{Arikan09}. 
  \end{algorithmic}  
\end{algorithm}
\begin{algorithm}
  \caption{Decoding}
  \label{alg:p}
   \begin{algorithmic}[1]
\REQUIRE $(R_b)_{b \in \mathcal{B}}$, $(f^{\textbf{SC}}_2(\widetilde{A}^{1:K}_{b,l}))_{l\in \mathcal{L}, b\in \mathcal{B}}$,  $(f^{\textbf{SC}}_1(\widetilde{A}^{1:K}_{B,l}))_{l\in \mathcal{L}}$\vspace{0.1em}
\STATE Define $\widehat{A}^{1:K}_{B,l} [\mathcal{V}_{U|Y}]\triangleq f^{\textbf{SC}}_1(\widetilde{A}^{1:K}_{B,l})$ for any $l \in \mathcal{L}$ 
              \FOR{Block $b \in \mathcal{B}$ from $b=B$ to $b=1$}
    \FOR{$l \in \mathcal{L}$}
        \STATE Form an estimate of $\widetilde{A}^{1:K}_{b,l}$ as $$\widehat{A}^{1:K}_{b,l} \triangleq g^{\textup{SC}}(\widehat{A}^{1:K}_{b,l}[\mathcal{V}_{U|Y}],f^{\textbf{SC}}_2(\widetilde{A}^{1:K}_{b,l}), \widetilde{Y}^{1:K}_{b,l})$$
	\ENDFOR
	\STATE From Line 7 in Algorithm~\ref{alg:p1}, determine an estimate of $T^{1:|\mathcal{V}_U|L}_b$ as $$\widehat{T}^{1:|\mathcal{V}_U|L}_b \triangleq \displaystyle\concat_{l\in \mathcal{L}} \widehat{A}^{1:K}_{b,l}[\mathcal{V}_U]$$ \vspace{-1em}
	\STATE  From Line 4 in Algorithm~\ref{alg:p1}, form   an estimate of $({M}_b \Vert {M'_b} \lVert{R'_b})$ as
$$
(\widehat{M}_b \Vert \widehat{M'_b} \lVert\widehat{R'_b}) \triangleq R_b \odot \widehat{T}^{1:|\mathcal{V}_U|L}_b 
$$ \vspace{-1em}
		\STATE From Line 3 in Algorithm~\ref{alg:p1} and $ \widehat{M'_b}$, form $\left(\widehat{A}^{1:K}_{b-1,l}[\mathcal{V}_{U|Y}]\right)_{l\in\mathcal{L}}$  an estimate of $ \left( f^{\textbf{SC}}_1 \left( \widetilde{A}^{1:K}_{b-1,l} \right) \right)_{l\in \mathcal{L}}$  	\ENDFOR
  \end{algorithmic}  
\end{algorithm}
\begin{figure}[h]
\centering
  \includegraphics[width=8.5 cm]{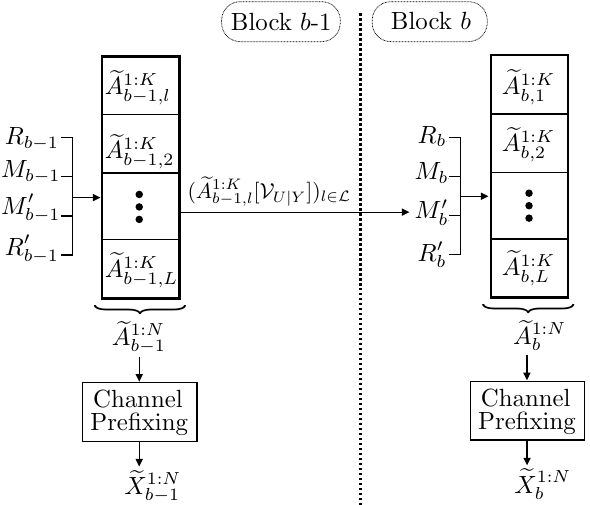}
  \caption{In Block $b \in \mathcal{B}$, $\widetilde{A}^{1:N}_b$ is made of $L$ sub-blocks $(\widetilde{A}^{1:K}_{b,l})_{l \in \mathcal{L}}$, which are constructed from $R_b$ (randomness for universal hashing), $M_b$ (secret message), $M_b'$ (a part of $\widetilde{A}^{1:N}_{b-1}$ from Block $b-1$), and $R_b'$ (local randomness). The construction of $M_b'$ in Line 3 of Algorithm \ref{alg:p1} creates a dependency between Block $b \in \llbracket 2, B \rrbracket$ and Block $b-1$. In Block $b \in \mathcal{B}$, the codeword $\widetilde{X}^{1:N}_b$, to be sent over the channel, is then obtain via channel prefixing from~$\widetilde{A}^{1:N}_b$.}
  \label{fig}
\end{figure}

\begin{figure}
\centering
  \includegraphics[width=8.5 cm]{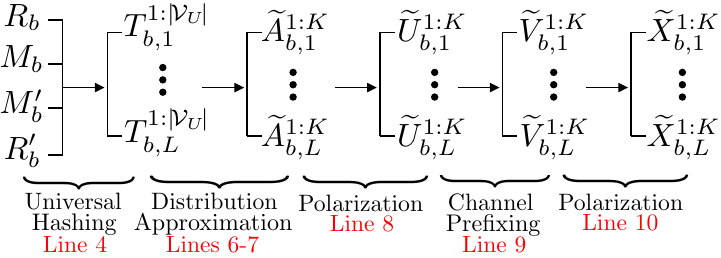}
  \caption{Summary of the steps in Algorithm \ref{alg:p1} to obtain the codeword $\widetilde{X}^{1:N}_b \triangleq \concat_{l\in \mathcal{L}} (\widetilde{X}_{b,l}^{1:K})$ from $R_b$ (randomness for universal hashing), $M_b$ (secret message), $M_b'$ (a part of Block $b-1$), and $R_b'$ (local randomness). Line 7 (distribution approximation) describes the creation of the $L$ sub-blocks $(\widetilde{A}_{b,l}^{1:K})_{l\in\mathcal{L}}$ from $T^{1:|\mathcal{V}_U|L}_b$,  and ensures that their distribution is close to the product distribution $q_{A^{1:N}}$, which will be a crucial fact to analyze the information leakage of the coding~scheme.}
  \label{fig2}
\end{figure}
\textbf{High-level description of the coding scheme}:
We depict in Figure \ref{fig} how codewords are created at the transmitter. Note that there exists an interdependence between two consecutive encoding blocks since $M'_{b}$, $b\in \llbracket 2,B \rrbracket$, used in Block $b$, is obtained from Block $b-1$, as described in Line~3 of~Algorithm~\ref{alg:p1}. 

Consider Block $b\in \mathcal{B}$ of Algorithm~\ref{alg:p1}. The encoder starts by creating $T^{1:|\mathcal{V}_U|L}_b$ via universal hashing applied on  the sequence created by $M'_b$, the secret message $M_b$, and the local randomness $R_b'$, as described in Line 4. Next,  $T^{1:|\mathcal{V}_U|L}_b$ is broken down into $L$ pieces with same length in Line 6, from which the encoder creates $L$ sub-blocks $(\widetilde{A}_{b,l}^{1:K})_{l\in\mathcal{L}}$, as described   in Line 7. Then, from $(\widetilde{A}_{b,l}^{1:K})_{l\in\mathcal{L}}$, the codewords $(\widetilde{X}_{b,l}^{1:K})_{l\in\mathcal{L}}$, are obtained via channel prefixing, as described   in Lines~8-10. The codewords $(\widetilde{X}_{b,l}^{1:K})_{l\in\mathcal{L}}$ are sent over the channel and their noisy observations at the legitimate receiver are denoted by $( \widetilde{Y}_{b,l}^{1:K})_{l\in\mathcal{L}}$.
Note that the $L$ sub-blocks $(\widetilde{A}_{b,l}^{1:K})_{l\in\mathcal{L}}$ are created such that their distribution is close to the product distribution $q_{A^{1:N}}$. A crucial point to ensure this property comes from the uniformity of $M'_b$, i.e., the uniformity of  $\widetilde{A}^{1:K}_{b-1,l}[\mathcal{V}_{U|Y}]$, ${l\in \mathcal{L}}$, which follows from Line~7 and the property $\mathcal{V}_{U|Y} \subset \mathcal{V}_U$.    Finally, as described in Line~15, using a pre-shared secret (obtained from the initialization phase in Section \ref{sec:initk}), the encoder applies a one-time pad to $(\widetilde{A}^{1:K}_{b,l}[\mathcal{H}_{U|Y} \backslash \mathcal{V}_{U|Y}])_{l\in \mathcal{L}, b\in \mathcal{B}}$, and $(\widetilde{A}^{1:K}_{B,l}[\mathcal{V}_{U|Y}])_{l\in \mathcal{L}}$, and sends the result to the legitimate receiver with a channel code \cite{Arikan09}. This step is done for technical reasons: $(\widetilde{A}^{1:K}_{b,l}[\mathcal{H}_{U|Y} \backslash \mathcal{V}_{U|Y}])_{l\in \mathcal{L}, b\in \mathcal{B}}$ are not uniformly distributed and could not be included in the definition of $M'_b$, $b\in \mathcal{B}$, as our analysis relies on the uniformity of $M'_b$, $b\in \mathcal{B}$. However, as shown later, the length of $(\widetilde{A}^{1:K}_{b,l}[\mathcal{H}_{U|Y} \backslash \mathcal{V}_{U|Y}])_{l\in \mathcal{L}, b\in \mathcal{B}}$ is negligible compared to $NB$ such that the overall communication rate is not affected. It will also be shown that this has a negligible effect on the overall information leakage to the~eavesdropper.

In a given block  $b\in \mathcal{B}$, we depict in Figure \ref{fig2} a summary of the different phases in Algorithm~\ref{alg:p1} through which the encoder output  is obtained from the local randomness $R_b'$, the secret message $M_b$, the randomness $R_b$ used for universal hashing, and $M'_b$. Note that $R_b$ needs to be shared between the legitimate users but does not need to be secret from the eavesdropper, and can be recycled over several blocks so that the exchange of necessary randomness for universal hashing between the legitimate users does not affect the overall communication rate.

At the decoder, the legitimate receiver first estimates $(\widetilde{A}_{B,l}^{1:K})_{l\in\mathcal{L}}$ from  $(\widetilde{A}^{1:K}_{B,l}[\mathcal{H}_{U|Y}])_{l\in \mathcal{L}}$ and $( \widetilde{Y}_{B,l}^{1:K})_{l\in\mathcal{L}}$, as described in Lines 2-4 of Algorithm \ref{alg:p} for Block $B$. Then, from this estimate of $(\widetilde{A}_{B,l}^{1:K})_{l\in\mathcal{L}}$,  the legitimate receiver forms an estimate of $M_B$ and $M_{B'}$, as described in Lines 6-7 of Algorithm \ref{alg:p} for Block $B$. Next, to estimate the message $M_{B-1}$ and $M_{B-1}'$, the legitimate receiver uses the estimate of $M_{B'}$ along with $(\widetilde{A}^{1:K}_{B-1,l}[\mathcal{H}_{U|Y} \backslash \mathcal{V}_{U|Y}])_{l\in \mathcal{L}}$, and $( \widetilde{Y}_{B-1,l}^{1:K})_{l\in\mathcal{L}}$, as described in Lines 2-7 of Algorithm~\ref{alg:p} for Block $B-1$. 
 Hence, the legitimate receiver can estimate all the messages $(M_b)_{b\in\mathcal{B}}$ starting from the last block and iterating through the previous blocks via the loop in Line 2 of Algorithm \ref{alg:p}.

 Note that in the analysis of the coding scheme secrecy rate, one needs to account for $(i)$ the one-time pad in Line~15 of Algorithm~\ref{alg:p1}, $(ii)$ the transmission  of the randomness $(R_b)_{1:B}$ that is used in Algorithms \ref{alg:p1} and \ref{alg:p}, and $(iii)$ the initialization phase (Algorithms \ref{alg:p3} and \ref{alg:p4}). We will show that $(i)$, $(ii)$, and $(iii)$ are done with a negligible impact on the secrecy rate in Sections \ref{secneg1}, \ref{secneg2}, and \ref{sec:impactk}, respectively.

\section{Proof of Theorem \ref{thmain1} with a pre-shared  key} \label{sec:proof1a}
In this section, we prove Theorem \ref{thmain1} when the legitimate users have access to a pre-shared secret key whose rate is negligible. Hence, we ignore in this section the initialization phase, i.e., Algorithms~\ref{alg:p3},~\ref{alg:p4}. We also assume in this section that  all the components of $\mathbf{s}_b$, $b \in \mathcal{B}$, are identical and equal to $s$. To simplify notation, we write $s$ instead of $\mathbf{s}_b$, $b\in \mathcal{B}$.   
\subsection{Characterization of the distribution induced by the encoder}
Let $\widetilde{p}_{U_{b}^{1:N}X_{b}^{1:N}Y_{b}^{1:N}Z_{b}^{1:N}(s)}$ denote the distribution induced by the encoding scheme described in Algorithm~\ref{alg:p1}. Lemma \ref{lemdist} gives an approximation of $\widetilde{p}_{U_{b}^{1:N}X_{b}^{1:N}Y_{b}^{1:N}Z_{b}^{1:N}(s)}$ in  terms of the distribution $q_{UXYZ(s)}$ defined in Section \ref{sec:notation}. This result will be useful in our subsequent analysis.
\begin{lem} \label{lemdist}
For $b\in \mathcal{B}$, we have
\begin{align*}
 \mathbb{D}( q_{U^{1:N}X^{1:N}Y^{1:N}Z^{1:N}(s)} \lVert \widetilde{p}_{U_{b}^{1:N}X_{b}^{1:N}Y_{b}^{1:N}Z_{b}^{1:N}(s)} )
 \leq 2LK \delta_K,
\end{align*}
where $ q_{U^{1:N}X^{1:N}Y^{1:N}Z^{1:N}(s)} \triangleq \prod_{i=1}^Nq_{UXYZ(s)}$.
\end{lem}
\begin{proof}
See Appendix \ref{App_lemdist}.
\end{proof}

\subsection{Reliability} \label{sec:reliability}
We now show that the receiver is able to recover the original message with a vanishing error probability. Define $\widehat{M}_{1:B} \triangleq (\widehat{M}_b)_{b\in \mathcal{B}}$. Define for $b \in \mathcal{B}$, $ \widehat{A}^{1:N}_{b} \triangleq \concat_{l\in \mathcal{L}} \widehat{A}_{b,l}^{1:K}$, $ \widetilde{A}^{1:N}_{b} \triangleq \concat_{l\in \mathcal{L}} \widetilde{A}_{b,l}^{1:K}$, $ {A}^{1:N}_{b} \triangleq {A}^{1:LK}$, $\mathcal{E}_{b-1} \triangleq \{ \widehat{A}^{1:N}_{b} \neq  \widetilde{A}_{b}^{1:N} \}$, and $\mathcal{E}_{A_b} \triangleq \{ (\widetilde{Y}_{b}^{1:N},\widetilde{A}_{b}^{1:N}) \neq ({Y}_{b}^{1:N},{A}_{b}^{1:N}) \}$. For $b \in \mathcal{B}$, consider a coupling~\cite[Lemma~3.6]{Aldous83} between $\widetilde{p}_{{Y}_{b}^{1:N}{A}_{b}^{1:N}}$ and $q_{{Y}_{b}^{1:N}{A}_{b}^{1:N}}$ such that $\mathbb{P} [\mathcal{E}_{A_b}] = \mathbb{V}(\widetilde{p}_{{Y}_{b}^{1:N}{A}_{b}^{1:N}} ,q_{{Y}_{b}^{1:N}{A}_{b}^{1:N}})$. For $b \in \mathcal{B}$,  consider $({A}_{b}^{1:N},{Y}_{b}^{1:N},\widetilde{A}_{b}^{1:N},\widetilde{Y}_{b}^{1:N})$ distributed according to this coupling, then 
\begin{align}
&\mathbb{P}\left[\widehat{M}_{1:B} \neq M_{1:B} \right] \nonumber \\ \nonumber
& \leq \sum_{b\in\mathcal{B}} \mathbb{P}\left[\widehat{M}_b \neq M_b\right]  \displaybreak[0]\\ \nonumber
& \stackrel{(a)}{\leq} \sum_{b\in\mathcal{B}} \mathbb{P}\left[\widehat{T}^{1:|\mathcal{V}_U|L}_b \neq {T}^{1:|\mathcal{V}_U|L}_b\right]  \displaybreak[0]\\ \nonumber
&\stackrel{(b)}{\leq} \sum_{b\in\mathcal{B}} \mathbb{P}\left[ \widehat{A}^{1:N}_{b} \neq  \widetilde{A}^{1:N}_{b} \right] \displaybreak[0]\\ \nonumber
& \leq \sum_{b\in\mathcal{B}}  \left[  \mathbb{P}\left[ \widehat{A}^{1:N}_{b} \neq \widetilde{A}^{1:N}_{b}|\mathcal{E}_{A_b}^c \cap \mathcal{E}_{b}^c \right] + \mathbb{P} [\mathcal{E}_{A_b} \cup \mathcal{E}_{b}] \right] \nonumber \displaybreak[0] \\
& \stackrel{(c)}{\leq} \sum_{b\in\mathcal{B}}  \left[ \sum_{l\in\mathcal{L}} \mathbb{P}\left[ \widehat{A}^{1:K}_{b,l} \neq \widetilde{A}^{1:K}_{b,l}|\mathcal{E}_{A_b}^c \cap \mathcal{E}_{b}^c \right] + \mathbb{P} [\mathcal{E}_{A_b}] + \mathbb{P} [ \mathcal{E}_{b}] \right] \nonumber \\
& \stackrel{(d)}{\leq}  \sum_{b\in\mathcal{B}} \left[ KL \delta_K + \sqrt{2 \ln 2} \sqrt{2LK \delta_K}  + \mathbb{P}\left[ \widehat{A}^{1:N}_{b+1} \neq \widetilde{A}^{1:N}_{b+1} \right]  \right]\nonumber \\
& \stackrel{(e)}{\leq}  \sum_{b\in\mathcal{B}} \left[ ( KL \delta_K + \sqrt{2 \ln 2} \sqrt{2LK \delta_K}) (B-b +1)  \right] \nonumber \\
& = ( KL \delta_K +\sqrt{2 \ln 2} \sqrt{2LK \delta_K}) B(B+1)/2 , \label{eq:reliability}  
\end{align}
where $(a)$ holds by Line 7 in Algorithm \ref{alg:p}, $(b)$ holds by Line~6 in Algorithm \ref{alg:p}, $(c)$ holds by the union bound, $(d)$ holds because $\mathbb{P}\left[ \widehat{A}^{1:K}_{b,l} \neq \widetilde{A}^{1:K}_{b,l}|\mathcal{E}_{A_b}^c \cap \mathcal{E}_{b}^c \right] \leq K \delta_K$ by \eqref{lemscs} and because $\mathbb{P} [\mathcal{E}_{A_b}] =  \mathbb{V}(\widetilde{p}_{{Y}_{b}^{1:N}{A}_{b}^{1:N}} ,q_{{Y}_{b}^{1:N}{A}_{b}^{1:N}}) \leq \sqrt{2 \ln 2} \sqrt{2LK \delta_K}$ by Lemma \ref{lemdist} and Pinsker's inequality, $(e)$~holds by induction.

\subsection{Pre-shared key rate} \label{secneg1}
The coding scheme described in Algorithms \ref{alg:p1} and \ref{alg:p} involves a one-time pad to securely transmit $(f^{\textbf{SC}}_2(\widetilde{A}^{1:K}_{b,l}))_{l\in \mathcal{L}, b\in \mathcal{B}}$, and $(f^{\textbf{SC}}_1(\widetilde{A}^{1:K}_{B,l}))_{l\in \mathcal{L}}$, which requires a pre-shared key with length $l_{\textup{OTP}} \triangleq  LB|\mathcal{H}_{U|Y} \backslash \mathcal{V}_{U|Y}|  + L | \mathcal{V}_{U|Y}|$ and rate 
\begin{align*}
	\frac{l_{\textup{OTP}}}{NB} 
	& =  \frac {|\mathcal{H}_{U|Y}| -|\mathcal{V}_{U|Y}|}{K}   + \frac{| \mathcal{V}_{U|Y}|}{KB}  \displaybreak[0]\\
		& \leq  \frac {|\mathcal{H}_{U|Y}| -|\mathcal{V}_{U|Y}|}{K}   + \frac{1}{B} \displaybreak[0] \\
	& = \delta(K) + 1/B,
\end{align*}
where $\delta(K)$ is such that $\lim_{K \to \infty} \delta(K) =0$ since $\lim_{K \to \infty} |\mathcal{H}_{U|Y}| / K = H(U|Y)$ \cite{Arikan10}, and $\lim_{K \to \infty} |\mathcal{V}_{U|Y}| / K = H(U|Y)$ \cite{Honda13,Chou14rev}. 

\subsection{Blockwise Security Analysis} \label{sec:secindiv}

We prove in this section that security holds in each block $b \in \mathcal{B}$ individually. We use a series of lemmas to obtain this result and determine acceptable values for the parameter $r$  defined in \eqref{eqparamr}. For $(X,Z)$ distributed according to $p_{XZ}$, defined over the finite alphabet $\mathcal{X}\times \mathcal{Z}$, recall that the $\epsilon$-smooth min-entropy of $X$ given $Z$ is defined as \cite{renner2008security}
\begin{align*}
H_{\infty}^{\epsilon}(p_{XZ}|p_Z) \triangleq \max_{r_{XZ} \in \mathcal{B}^{\epsilon}(p_{XZ})} \min_{  z \in \textup{Supp}(p_{Z})} \min_{ x\in \mathcal{X}} \log   \frac{p_Z(z)}{r_{XZ}(x,z)} , 
\end{align*}
where $\textup{Supp}(p_Z) \triangleq \{ z \in \mathcal{Z}: p_Z(z) >0\}$ and $\mathcal{B}^{\epsilon}(p_{XZ}) \triangleq \{ (r_{XZ}:\mathcal{X} \times \mathcal{Z} \to [0,1]) : \mathbb{V} (p_{XZ},r_{XZ}) \leq \epsilon \}$.
We will also need the following version of the leftover hash lemma. 
\begin{lem}[\!\cite{renner2008security}]\label{lemamp}
Let $T$ and $Z$ be distributed according to $p_{TZ}$ over $\mathcal{T} \times \mathcal{Z}$. Consider $F : \mathcal{R} \times \{0,1\}^{k} \rightarrow \{0,1\}^{r}$, where the first input, denoted by $R$, is uniformly distributed over $\mathcal{R}$ to indicate that $F$ is chosen uniformly at random in a family  of two-universal hash functions.   Then, for any $\epsilon \in [0,1[$,  
\begin{align} \label{eq:lohl}
\mathbb{V} ( p_{F(R,T),R,Z}, p_{U_{\mathcal{K}}} p_{U_{\mathcal{R}}}p_Z)  
&\leq  2\epsilon + \sqrt{ 2^{ r - {H}_{\infty}^{\epsilon}\left( p_{TZ}| p_{Z }\right)} 
}, 
\end{align}
where $p_{U_{\mathcal{K}}}$ and $p_{U_{\mathcal{R}}}$ are the uniform distribution over $\{0,1\}^{r}$ and $\mathcal{R}$, respectively.
\end{lem}
We now would like to use Lemma~\ref{lemamp} to make $(M_b \lVert M'_b)$ almost independent from the eavesdropper channel observations. However, in the encoding scheme described in Algorithm \ref{alg:p1}, $(M_b \lVert M'_b)$ is not defined as the output of a two-universal hash function as required in Lemma \ref{lemamp}. To overcome this challenge, we show in the following lemma  that the distribution $\widetilde{p}$ induced by the encoder in Algorithm \ref{alg:p1} also describes a process for which $(M_b \lVert M'_b)$ is  defined as $(M_b \lVert M'_b) \triangleq g^{\textbf{UH}}_{R_b}(T_b^{1:|\mathcal{V}_U|L} ,r )$ where $r$ is defined in~\eqref{eqparamr}. For convenience, we write in the following $F(R_b,T_b^{1:|\mathcal{V}_U|L}) \triangleq g^{\textbf{UH}}_{R_b}(T_b^{1:|\mathcal{V}_U|L} ,r )$.

\begin{lem} \label{lemint}
Fix $b\in \mathcal{B}$. To simplify notation, we write $T_b$ instead of $T_b^{1:|\mathcal{V}_U|L}$, $\widetilde{Z}_b(s)$ instead of $\widetilde{Z}_b^{1:N}(s)$, $\widetilde{X}_b$ instead of $\widetilde{X}_b^{1:N}$, and $Z_b(s)$ instead of $Z_b^{1:N}(s)$. We also define $\bar{M}_b \triangleq (M_b \lVert M'_b)$ such that $T_b \triangleq R_b^{-1} \odot (\bar{M}_b \lVert R_b')$. Next, define
\begin{align}
\widetilde{q}_{\bar{M}_b T_b  X_b  Z_b(s)    R_b } \triangleq \widetilde{p}_{ X_b Z_b(s)   |T_b  } \widetilde{q}_{T_b  } \widetilde{q}_{R_b}  \widetilde{q}_{\bar{M}_b|T_b   R_b},
\end{align}
with $\widetilde{q}_{T_b  }$ the uniform distribution over $\{0,1\}^{|\mathcal{V}_U|L}$, $ \widetilde{q}_{R_b}$ the uniform distribution over $\mathcal{R}$,  and $\forall \bar{m}_b, \forall t_b  , \forall r_{b}$, $\widetilde{q}_{\bar{M}_b|T_b   R_b} ( \bar{m}_b|t_b,  r_{b}) \triangleq  \mathds{1} \{ \bar{m}_b = F(r_b,t_b  ) \}$. Then, we~have  $$\widetilde{p}_{\bar{M}_b T_b X_b Z_b(s)    R_b} = \widetilde{q}_{\bar{M}_b T_b   X_b Z_b(s)     R_b}.$$
\end{lem}
\begin{proof}
See Appendix \ref{App_lemimt}.
\end{proof}

Let $\mathcal{A}_b \subset \llbracket 1, N \rrbracket$ such that $|\mathcal{A}_b| = \alpha N$ and consider $\widetilde{X}^{1:N}_b[\mathcal{A}_b]$, the $\alpha N$ symbols that the eavesdropper has chosen to have access to in Block $b\in \mathcal{B}$.  
We study, by combining Lemmas~\ref{lemamp},~\ref{lemint},  the independence between $(R_b,\widetilde{Z}_b^{1:N}(s),\widetilde{X}^{1:N}_b[\mathcal{A}_b])$, i.e., all the knowledge  at the eavesdropper in Block $b\in \mathcal{B}$, and $(M_b \lVert M'_b)$ as follows.

\begin{lem} \label{lemampapp}
Fix $b\in \mathcal{B}$. We adopt the same notation as in Lemma \ref{lemint} and also write $\widetilde{X}_b[\mathcal{A}_b]$ instead of $\widetilde{X}^{1:N}_b[\mathcal{A}_b]$ for convenience. We have for any $\gamma \in]0,1[$
\begin{align}
& \mathbb{V} ( \widetilde{p}_{\bar{M}_{b} R_{b} Z_b(s) {X}_b[\mathcal{A}_b] }, \widetilde{p}_{\bar{M}_{b}} \widetilde{p}_{R_{b}Z_b (s){X}_b[\mathcal{A}_b] })  \nonumber \displaybreak[0] \\
& \leq 2^{1-L^{\gamma}} +  \sqrt{ 2^{ r -  {H} \left(T_{b}| \widetilde{Z}_{b} (s) \widetilde{X}_b[\mathcal{A}_b]\right) + N \delta^{(1)}(K,L)  }} , \label{eqam}
\end{align}
where $\delta^{(1)}(K,L) \triangleq  (K^{-1}+1)\sqrt{2 L^{\gamma -1}} $. 
\end{lem}
\begin{proof}
See Appendix \ref{App_lemampapp}.
\end{proof}

Next, using Lemma \ref{lemdist}, we lower bound the conditional entropy in \eqref{eqam} in the following lemma.
\begin{lem}\label{lemcond}
Fix $b\in \mathcal{B}$. We adopt the same notation as in Lemmas \ref{lemint}, \ref{lemampapp}. We have
\begin{align*}
&H \left(T_{b}| \widetilde{Z}_{b}(s) \widetilde{X}_b[\mathcal{A}_b]  \right) \\
&\geq N[(1-\alpha) H(U|Z(s) ) +\alpha H(U|X) - \delta^{(2)}(K,L)], 
\end{align*}
with $ \delta^{(2)}(K,L) \triangleq  2 \sqrt{2 \ln2} \sqrt{2N \delta_K} ( \log (|\mathcal{X}|^2 \max_{s \in \mathfrak{S}}|\mathcal{Z}_s| )  - N^{-1} \log (\sqrt{2 \ln2} \sqrt{2N \delta_K}))+ N^{-1}{H}_b ( N \delta_K  ) + N \delta_K   + o(1)$, and $H_b(\cdot)$ the binary entropy.
\end{lem}
\begin{proof}
See Appendix \ref{App_lemcond}.
\end{proof}

By combing Lemma \ref{lemampapp} and Lemma \ref{lemcond} we obtain the following result.
\begin{lem} \label{lemampapp2}
Fix $b\in \mathcal{B}$. We adopt the same notation as in Lemma \ref{lemcond}. 
We have for any $\gamma \in]0,1[$
\begin{align*}
& \mathbb{V} ( \widetilde{p}_{\bar{M}_{b} R_{b}Z_b(s) {X}_b[\mathcal{A}_b] }, \widetilde{p}_{\bar{M}_{b}} \widetilde{p}_{R_{b}Z_b (s){X}_b[\mathcal{A}_b] }) 
\\
& \leq 2^{1-L^{\gamma}} +  \sqrt{ 2^{ r -  N [(1-\alpha) H(U|Z(s) ) +\alpha H(U|X) - \delta^{(3)}(K,L) ]  }} , 
\end{align*}
where $\delta^{(3)}(K,L) \triangleq \delta^{(1)}(K,L) + \delta^{(2)}(K,L)$, with $\delta^{(1)}(K,L)  $ defined in Lemma \ref{lemampapp} and $\delta^{(2)}(L,K)$ defined in Lemma~\ref{lemcond}.
\end{lem}

Finally, we obtain security in a given block as follows.
\begin{lem} \label{lemsecurblockwise}
Fix $b\in \mathcal{B}$ and $\xi>0$. We choose $$r \triangleq \! N \!\!\left[ (1 \! -\! \alpha) \min_{s\in \mathfrak{S}} H ({U}| {Z}(s) )\! +\! \alpha H(U|X) \! -\!  \delta^{(3)}(K,L) \!- \!\xi \right]$$ with $\delta^{(3)}(K,L)$ defined in Lemma \ref{lemampapp2}. Then, for $L$ large enough
\begin{align*}
 I\left(M_{b}M'_{b}; \widetilde{Z}_{b}(s) {X}_b[\mathcal{A}_b]R_b\right) \leq  \delta^{(4)}(K,L,\xi),
\end{align*}
where $\delta^{(4)}(K,L,\xi) \triangleq (2^{1-L^{\gamma}} + \sqrt{ 2^{ -N \xi }}) \log \frac{2^N}{2^{1-L^{\gamma}} + \sqrt{ 2^{ -N \xi }} }$ .
\end{lem}
\begin{proof}
We adopt the same notation as in the previous lemmas. By definition of $r$ and by Lemma \ref{lemampapp2}, we have  
\begin{align}
\mathbb{V} ( \widetilde{p}_{\bar{M}_{b}R_{b}Z_b(s) {X}_b[\mathcal{A}_b] }, \widetilde{p}_{\bar{M}_{b}} \widetilde{p}_{R_{b}Z_b (s){X}_b[\mathcal{A}_b] }) 
 \leq 2^{1-L^{\gamma}} + \sqrt{ 2^{ -N \xi }}. \label{eqintermed}	
\end{align}
We thus have 
\begin{align}
&I(M_{b}M'_{b}; \widetilde{Z}_{b}(s) {X}_b[\mathcal{A}_b] R_b) \nonumber  \\ \nonumber
& = I(\bar{M}_{b}; \widetilde{Z}_{b}(s){X}_b[\mathcal{A}_b] R_b) 
\\ \nonumber
& \stackrel{(a)}{\leq} f(\mathbb{V} ( \widetilde{p}_{\bar{M}_{b}R_{b}Z_b(s) {X}_b[\mathcal{A}_b] }, \widetilde{p}_{\bar{M}_{b}} \widetilde{p}_{R_{b}Z_b (s){X}_b[\mathcal{A}_b] })  ) \displaybreak[0]\\
& \stackrel{(b)}{\leq}  f(2^{1-L^{\gamma}} + \sqrt{ 2^{ -N \xi }}), \label{eqsmax}
\end{align}
where $(a)$ holds by \cite[Lemma 2.7]{bookCsizar} with $f:x \mapsto x \log (2^N/x)$, $(b)$ holds for $L$ large enough   since $f$ is increasing for small enough~values. 
\end{proof}

\subsection{Analysis of security over all blocks jointly } \label{sec:jointsec}

We obtain security over all blocks jointly from Lemma \ref{lemsecurblockwise} as follows.
\begin{lem} \label{lemjointsec}
For convenience, we define for $i,j \in \mathcal{B}$, $\widetilde{Z}_{1:i}(s)\triangleq (\widetilde{Z}^{1:N}_{b}(s))_{b\in \llbracket 1,i \rrbracket}$, $\widetilde{X}_{1:i}[\mathcal{A}]\triangleq (\widetilde{X}^{1:N}_{b}[\mathcal{A}_b])_{b\in \llbracket 1,i \rrbracket}$, $R_{i:j} \triangleq (R_b)_{b\in \llbracket i,j \rrbracket}$, and $M_{i:j} \triangleq (M_b)_{b\in \llbracket i,j \rrbracket}$. We have 
\begin{align*}
\max_{s\in \mathfrak{S}} \max_{\mathcal{A} \in \mathbb{A}} I(M_{1:B}; \widetilde{Z}_{1:B}(s)\widetilde{X}_{1:B}[\mathcal{A}] R_{1:B})
\leq 2B \delta^{(4)}(L,K,\xi), 
\end{align*}
where $\delta^{(4)}(L,K,\xi)$ is defined in Lemma~\ref{lemsecurblockwise}. 
\end{lem}
\begin{proof}
For convenience, define for $i \in \mathcal{B}$, $L_{i} \triangleq  (\widetilde{Z}_{i}(s),\widetilde{X}_{i}[\mathcal{A}_{i}],R_{i})$ and $L_{1:i} \triangleq  (\widetilde{Z}_{1:i}(s), \widetilde{X}_{1:i}[\mathcal{A}], R_{1:i})$. Then, 
\begin{align}
&I(M_{1:B};L_{1:B}) \nonumber \\ \nonumber
& = \sum_{i=0}^{B-1}  I(M_{1:B}; L_{i+1}| L_{1:i}) \displaybreak[0]\\ \nonumber
&\stackrel{(a)}{=} \sum_{i=0}^{B-1}  I(M_{1:i+1}; L_{i+1}| L_{1:i}) \displaybreak[0] \\ \nonumber
&\leq \sum_{i=0}^{B-1}  I(M_{1:i+1}L_{1:i} ;  L_{i+1})\displaybreak[0] \\ \nonumber
& = \sum_{i=0}^{B-1}  I(M_{i+1};  L_{i+1}) + I(M_{1:i} L_{1:i} ;  L_{i+1}|M_{i+1}) \displaybreak[0] \\ \nonumber
& \stackrel{(b)}{\leq} B \delta^{(4)}(K,L,\xi) +  \sum_{i=0}^{B-1}  I(M_{1:i} M_{i+1}' L_{1:i} ;L_{i+1} M_{i+1}) \displaybreak[0]\\ \nonumber
& \stackrel{(c)}{=} B \delta^{(4)}(K,L,\xi) +  \sum_{i=0}^{B-1}  I( M_{i+1}' ; L_{i+1} M_{i+1}) \displaybreak[0]\\ \nonumber
& \stackrel{(d)}{=}  B \delta^{(4)}(K,L,\xi) +  \sum_{i=0}^{B-1}  I( M_{i+1}' ; L_{i+1}|M_{i+1})\displaybreak[0] \\ \nonumber
& \leq  B \delta^{(4)}(K,L,\xi) +  \sum_{i=0}^{B-1}  I( M_{i+1}M_{i+1}' ; L_{i+1})\\
& \stackrel{(e)}{\leq} 2B \delta^{(4)}(K,L,\xi) ,\label{eqmaxj}
\end{align}
where  $(a)$ holds by the chain rule and since we have 
$
I(M_{i+2:B}; L_{i+1}| L_{1:i} M_{1:i+1}) \leq  I(M_{i+2:B}; L_{1:i+1} M_{1:i+1} ) = 0$, $(b)$ holds by Lemma~\ref{lemsecurblockwise}, $(c)$~holds by the chain rule and because $(M_{1:i},L_{1:i}) -  M_{i+1}' - (L_{i+1},M_{i+1})$ forms a Markov chain, $(d)$ holds by independence between $M_{i+1}'$ and $M_{i+1}$, $(e)$ holds by Lemma~\ref{lemsecurblockwise}. 
The lemma holds since \eqref{eqmaxj} holds for any $s\in \mathfrak{S}$ and any $\mathcal{A} \in \mathbb{A}$. 
\end{proof}
\subsection{Secrecy Rate}

The rate of the transmitted messages is
\begin{align*}
& \frac{\sum_{b\in\mathcal{B}}|M_b|}{BN} \displaybreak[0]  \stackrel{(a)}{=} \frac{r + (B-1) (r - L |\mathcal{V}_{U|Y}| )}{BN} \displaybreak[0]\\
& \geq \frac{r}{N} -  \frac{ |\mathcal{V}_{U|Y}| }{K} \displaybreak[0]\\
& \stackrel{(b)}{=} I(U;Y) - \alpha I(U;X) - (1-\alpha) \smash{\max_{s\in \mathfrak{S}}} I \left({U}; {Z}(s) \right) \\
& \phantom{--}   -  \delta^{(3)}(K,L) - \xi +o(1),
\end{align*}
where $(a)$ holds by \eqref{eqparamr}, $(b)$ holds by the choice of $r$ in Lemma~\ref{lemsecurblockwise} and because $\displaystyle\lim_{K \to \infty} |\mathcal{V}_{U|Y}| / K = H(U|Y)$ by~\cite{Chou14rev}. 

\subsection{Randomness amortization} \label{secneg2}
The randomness $(R_b)_{1:B}$ in the coding scheme of Section~\ref{secCS} needs to be shared between the legitimate users. This can be done with negligible impact on the overall communication rate similar to~\cite{Bellare2012} using an hybrid argument by repeating  the coding scheme of Section \ref{secCS} with the same randomness~$(R_b)_{1:B}$.

\section{Proof of Theorem \ref{thmain1} without pre-shared key} \label{sec:initphase}
The coding scheme of Section \ref{secCS} requires a pre-shared secret key between the legitimate users. We now consider the initialization phase, described in Algorithms \ref{alg:p3}, \ref{alg:p4}, to generate such a key with negligible impact on the overall communication rate. 
We study the reliability and the secrecy of the generated key in Sections \ref{sec:keyre} and \ref{sec:keysec}, respectively, the impact of the initialization phase on the overall communication rate in Section \ref{sec:impactk}, and the joint secrecy of the initialization phase and the coding scheme of Section \ref{secCS} in Section \ref{secjoint}. We adopt the same notation as in Section \ref{sec:proof1a}.

\subsection{Key reliability} \label{sec:keyre}
Similar to Lemma \ref{lemdist}, we have the following result. 
\begin{lem} \label{lemdistB}
For $b\in\mathcal{B}_0$, the distribution $\widetilde{p}$ induced by the encoder of Algorithm \ref{alg:p3} is approximated as follows.
\begin{align*}
\mathbb{D} ({q}_{U^{1:N}X^{1:N}Y^{1:N}Z^{1:N}(s)} \lVert  \widetilde{p}_{U^{1:N}_bX^{1:N}_bY_b^{1:N}Z_b^{1:N}(s)} ) \leq 2LK\delta_K.
\end{align*}
\end{lem}
 Then, we have
\begin{align*}
&\mathbb{P} \left[ (\widehat{\textup{Key}}_{b})_{b \in \mathcal{B}_0} \neq (\textup{Key}_{b})_{b \in \mathcal{B}_0}  \right] \\
& \leq \mathbb{P} [ (\widehat{U}^{1:N}_{b})_{b \in \mathcal{B}_0} \neq (\widetilde{U}^{1:N}_{b})_{b \in \mathcal{B}_0}  ] \\
& \leq B_0 L (\sqrt{2 \ln 2}\sqrt{2K \delta_K} + 2K\delta_K),
\end{align*}
where the last inequality holds similar to \eqref{eq:reliability}. 

\subsection{Key secrecy} \label{sec:keysec}

We first show secrecy in a given Block $b \in \mathcal{B}_0$. Let $\mathcal{A}_b \subset \llbracket 1, N \rrbracket$ such that $|\mathcal{A}_b| = \alpha N$ and consider $\widetilde{X}^{1:N}_b[\mathcal{A}_b]$, the $\alpha N$ symbols that the eavesdropper has chosen to have access to in Block $b\in \mathcal{B}_0$. Define ${p}_{\mathcal{U}_{\textup{Key}}}$ the uniform distribution over $\{ 0,1\}^{l'_{\textup{key}}}$. We have
\begin{align}
&\mathbb{V} ( \widetilde{p}_{\textup{Key}_{b}R_{b}^{\textup{init}} Z_b(s){X}_b[\mathcal{A}_b] D_b R^{\textup{init}'}_b }, {p}_{\mathcal{U}_{\textup{Key}}} \widetilde{p}_{R_{b}^{\textup{init}}Z_b(s) {X}_b[\mathcal{A}_b]D_b R_b^{\textup{init}'}})  \nonumber \\ \nonumber
& \stackrel{(a)}{\leq} 2\epsilon + \sqrt{ 2^{ l_{\textup{key}}' - {H}_{\infty}^{\epsilon} \left(\widetilde{p}_{U_{b} {Z}_{b}(s) {X}_b[\mathcal{A}_b] D_b R_b^{\textup{init}'}} | \widetilde{p}_{ {Z}_{b}(s) {X}_b[\mathcal{A}_b] D_b R_b^{\textup{init}'}} \right)  }} \\
& \stackrel{(b)}{\leq} 2 \cdot 2^{-L^{\gamma}} + \sqrt{ 2^{ l_{\textup{key}}' - {H} \left(\widetilde{U}_{b} | \widetilde{Z}_{b} (s) \widetilde{X}_b[\mathcal{A}_b] D_b R_b^{\textup{init}'} \right) + N \delta^{(1)}(K,L)  }} , \label{eqvardK}
\end{align}
where $(a)$ holds by Lemma \ref{lemamp}, $(b)$ holds by Lemma \ref{lemholenstein} with $\gamma \in ]0,1[$ as in the proof of Lemma \ref{lemampapp} with $\delta^{(1)}(K,L)$ defined in Lemma \ref{lemampapp}.

\begin{lem} \label{lementK}
For $b \in \mathcal{B}_0$, we have
\begin{align*}
&{H} \left(\widetilde{U}_{b} | \widetilde{Z}_{b} (s) \widetilde{X}_b[\mathcal{A}_b] D_b R_b^{\textup{init}'} \right) \\
& \geq \!N[I(U;Y) \!-\! \alpha I(U;X)\! -\! (1\!-\!\alpha) I({U}; {Z}(s))\!-\!\delta^{(5)}(K,L)],
\end{align*}
where $\delta^{(5)}(K,L) \triangleq 2  \sqrt{2 \ln2} \sqrt{2N \delta_K} ( \log (|\mathcal{X}|^2 \max_{s \in \mathfrak{S}}|\mathcal{Z}_s|)-N^{-1}\log (\sqrt{2 \ln2} \sqrt{2N\delta_K}) ) + o(1).$
\end{lem}

\begin{proof}
We have 
\begin{align*}
&{H} \left(\widetilde{U}_{b} | \widetilde{Z}_{b} (s) \widetilde{X}_b[\mathcal{A}_b] D_b R_b^{\textup{init}'} \right) \\
& = {H} \left(\widetilde{U}_{b} | \widetilde{Z}_{b} (s) \widetilde{X}_b[\mathcal{A}_b]  \right) - {I} \left(D_b R_b^{\textup{init}'};\widetilde{U}_{b}| \widetilde{Z}_{b} (s) \widetilde{X}_b[\mathcal{A}_b]  \right) \\
& \geq    {H} \left(\widetilde{U}_{b} | \widetilde{Z}_{b} (s) \widetilde{X}_b[\mathcal{A}_b]  \right) -  L  (|\mathcal{H}_{U|Y}| +|\mathcal{H}_{U|Y} \backslash \mathcal{V}_{U|Y}| )\\
& \stackrel{(a)}{\geq}    {H} \left(\widetilde{U}_{b} | \widetilde{Z}_{b} (s) \widetilde{X}_b[\mathcal{A}_b]  \right) - N H(U|Y) - o(KL)\\
& \stackrel{(b)}{\geq} N(1-\alpha) H(U|Z(s) ) +N\alpha H(U|X) - N H(U|Y) \\
& \phantom{--}- N\delta^{(5)}(K,L),
\end{align*}
where $(a)$ holds because $\lim_{K \to \infty} |\mathcal{H}_{U|Y}| / K = H(U|Y)$~\cite{Arikan10}, and $\lim_{K \to \infty} |\mathcal{V}_{U|Y}| / K = H(U|Y)$ \cite{Honda13,Chou14rev}, $(b)$ holds similar to the proof of Lemma \ref{lemcond}.
\end{proof}

Next, we choose 

\begin{align*}
&l_{\textup{key}}' \triangleq N[I(U;Y) - \alpha I(U;X) - (1-\alpha) \smash{ \max_{s\in \mathfrak{S}}} I \left({U}; {Z}(s) \right) \\
& \phantom{----}   - \delta^{(1)}(K,L)- \delta^{(5)}(K,L) - \xi],
\end{align*} with $\xi>0$. By \eqref{eqvardK} and Lemma \ref{lementK}, we obtain for $b \in \mathcal{B}_0$,
\begin{align}
& \mathbb{V} ( \widetilde{p}_{\textup{Key}_{b}R_{b}^{\textup{init}} Z_b(s){X}_b[\mathcal{A}_b] D_b R^{\textup{init}'} }, {p}_{\mathcal{U}_{\textup{Key}}} \widetilde{p}_{R_{b}^{\textup{init}}Z_b(s) {X}_b[\mathcal{A}_b]D_b R_b^{\textup{init}'}}) \nonumber \\
&   \leq 2 \cdot 2^{-L^{\gamma}} + \sqrt{ 2^{ -N \xi}} . \label{eqvardistKey}
\end{align}

\begin{lem} \label{lemsecurkey}
We have for $L$ large enough
\begin{align*}
 I\left( \textup{Key}_b; \widetilde{Z}_{b}(s)\widetilde{X}_b[\mathcal{A}_b]D_b R^{\textup{init}}_b R^{\textup{init}'}_b\right)& \leq   \delta^{(4)}(K,L,\xi),\\
\log |\mathcal{K}_b| - H( \textup{Key}_b)& \leq  \delta^{(4)}(K,L,\xi),
\end{align*}
with $\mathcal{K}_b \triangleq \{0,1\}^{l'_{\textup{key}}}$ and $\delta^{(4)}(K,L,\xi)$ defined in Lemma~\ref{lemsecurblockwise}.
\end{lem}
\begin{proof}
	The first inequality holds as the proof of Lemma~\ref{lemsecurblockwise} by using \eqref{eqvardistKey} in place of \eqref{eqintermed}. The second inequality holds by~\cite[Lemma~2.7]{bookCsizar} and \eqref{eqvardistKey}.
\end{proof}

By mutual independence of all the $B_0$ blocks of the initialization phase, we obtain from Lemma~\ref{lemsecurkey} the following result.

\begin{lem} \label{lemsecurkeyjoint}
Define $\textup{Key} \triangleq (\textup{Key}_b)_{b \in \mathcal{B}_0}$ and $\mathcal{K} \triangleq \mathcal{K}_b^{B_0}$. Let $\widetilde{Z}^{\textup{init}}(s)$ denote all the knowledge of the eavesdropper related to the initialization phase, i.e., $\widetilde{Z}^{\textup{init}}(s) \triangleq (\widetilde{Z}_{b}(s),\widetilde{X}_b[\mathcal{A}_b],D_b, R^{\textup{init}}_b, R^{\textup{init}'}_b)_{b \in \mathcal{B}_0}$. Then, for $K$ large~enough
\begin{align*}
\max_{s\in \mathfrak{S}} \max_{\mathcal{A} \in \mathbb{A}} I\left( \textup{Key} ; \widetilde{Z}^{\textup{init}}(s)\right)& \leq B_0 \delta^{(4)}(K,L,\xi), \displaybreak[0]\\
\log |\mathcal{K}| - H( \textup{Key})& \leq  B_0 \delta^{(4)}(K,L,\xi).
\end{align*}
\end{lem}

\subsection{Impact of the initialization phase on the overall communication rate} \label{sec:impactk}

The initialization phase requires $\rho NB_0$ channel uses, for some fixed $\rho \in \mathbb{N}$, to generate the secret key and transmit $(D_b,R_b^{\textup{init}},R_b^{\textup{init}'})_{b \in \mathcal{B}_0}$.
We choose $B_0$ such that  
 \begin{align*}
B_0 =\left\lceil \frac{  l_{\textup{OTP}}}{l_{key}'}\right\rceil ,
 \end{align*}
 where $l_{\textup{OTP}}= o(NB)$ represents the key length necessary to perform the one-time pad that appears in Algorithms \ref{alg:p3}, \ref{alg:p4}.
Hence, the impact of the initialization phase on the overall communication  is  
\begin{align} \label{eqimpact}
 \rho N B_0  
& < \rho N \left(1+ \frac{ l_{\textup{OTP}} }{ l_{key}'} \right) =  \rho \frac{ o(NB) }{ l_{key}'/N} = o(NB).
 \end{align}

We deduce from \eqref{eqimpact} that the communication rate of the coding scheme of Section \ref{secCS} and the initialization phase (considered jointly) is the same as the communication rate of the coding scheme of Section \ref{secCS} alone.

\subsection{Security of Algorithms \ref{alg:p1}, \ref{alg:p} and the initialization phase when considered jointly} \label{secjoint}

Let $M_{\textup{OTP}}$ be the sequence that needs to be secretly transmitted with a one-time pad in Algorithm~\ref{alg:p1}. Let $C \triangleq M_{\textup{OTP}} \oplus \textup{Key}$ be the encrypted version of $M_{\textup{OTP}}$ using $\textup{Key}$, obtained in the initialization phase. Let $\widetilde{Z}_{\mathcal{B}}(s) \triangleq (\widetilde{Z}_{1:B}(s),\widetilde{X}_{1:B}[\mathcal{A}], R_{1:B})$ denote all the observations of the eavesdropper related to the coding scheme of Section~\ref{secCS}, excluding $C$. Let $Z^{\textup{init}}(s)$, defined as in Lemma \ref{lemsecurkeyjoint}, denote all the observations of the eavesdropper related to the initialization phase.
The following lemma shows that strong secrecy holds for the coding scheme of Section \ref{secCS} and the initialization phase considered jointly.
\begin{lem}
We have
\begin{align*}
\max_{s\in \mathfrak{S},\mathcal{A} \in \mathbb{A}}\! I(M_{1:B} ;C \widetilde{Z}_{\mathcal{B}}(s) \widetilde{Z}^{\textup{init}}(s)) \leq  2(B\!+\!B_0) \delta^{(4)}(K,L,\xi),
\end{align*}
where $\delta^{(4)}(K,L,\xi)$ is defined in Lemma \ref{lemsecurblockwise}.
\end{lem}

\begin{proof}
We have
\begin{align}
&I(M_{1:B} ; C \widetilde{Z}_{\mathcal{B}}(s) \widetilde{Z}^{\textup{init}}(s)) \nonumber \\ \nonumber
& \stackrel{(a)}{=} I(M_{1:B} ; \widetilde{Z}_{\mathcal{B}}(s))  + I(M_{1:B};C | \widetilde{Z}_{\mathcal{B}}(s) \widetilde{Z}^{\textup{init}}(s)) \displaybreak[0]\\ \nonumber
& \leq I(M_{1:B} ; \widetilde{Z}_{\mathcal{B}}(s))  + I(M_{1:B} \widetilde{Z}_{\mathcal{B}}(s) \widetilde{Z}^{\textup{init}}(s) ;C ) \displaybreak[0]\\
& = I(M_{1:B} ; \widetilde{Z}_{\mathcal{B}}(s) )  + I( C; M_{1:B}\widetilde{Z}_{\mathcal{B}}(s)) \nonumber \displaybreak[0]\\
& \phantom{--}+ I( C;  \widetilde{Z}^{\textup{init}}(s) |M_{1:B} \widetilde{Z}_{\mathcal{B}}(s) ), \label{eqmotp1}
\end{align}
where $(a)$ holds by the chain rule and because $I(M_{1:B} ;  \widetilde{Z}^{\textup{init}}(s) | \widetilde{Z}_{\mathcal{B}}(s)) \leq I(M_{1:B} \widetilde{Z}_{\mathcal{B}}(s) ;  \widetilde{Z}^{\textup{init}}(s) ) = 0$. Next, we have 
\begin{align}
& I( C; M_{1:B} \widetilde{Z}_{\mathcal{B}}(s) )  \nonumber  \displaybreak[0] \\
 & \leq \log |\mathcal{K}| - H(C| M_{1:B} \widetilde{Z}_{\mathcal{B}}(s)  )\nonumber  \displaybreak[0] \\ \nonumber
 & \leq \log |\mathcal{K}| -  H(\textup{Key} \oplus  M_{\textup{OTP}} |  M_{\textup{OTP}}M_{1:B} \widetilde{Z}_{\mathcal{B}}(s) ) \displaybreak[0] \\  \nonumber
  & = \log |\mathcal{K}| -  H(\textup{Key}  |  M_{\textup{OTP} } M_{1:B} \widetilde{Z}_{\mathcal{B}}(s)   ) \displaybreak[0] \\   
  & = \log |\mathcal{K}| -  H(\textup{Key} ). \label{eqmotp2}
\end{align}
We also have
\begin{align}
&  I( C;  \widetilde{Z}^{\textup{init}}(s)|M_{1:B} \widetilde{Z}_{\mathcal{B}}(s) ) \nonumber \displaybreak[0]\\
  & \leq   I( C  M_{\textup{OTP}};  \widetilde{Z}^{\textup{init}}(s)|M_{1:B} \widetilde{Z}_{\mathcal{B}}(s) ) \nonumber \\ \nonumber
& =    I( \textup{Key}  M_{\textup{OTP}};  \widetilde{Z}^{\textup{init}}(s)|M_{1:B} \widetilde{Z}_{\mathcal{B}}(s) ) \\ \nonumber
& \stackrel{(b)}{=} I( \textup{Key} ;  \widetilde{Z}^{\textup{init}}(s)| M_{\textup{OTP}}M_{1:B} \widetilde{Z}_{\mathcal{B}}(s)   ) \\ \nonumber
& \leq  I( \textup{Key} M_{\textup{OTP}}M_{1:B} \widetilde{Z}_{\mathcal{B}}(s);  \widetilde{Z}^{\textup{init}}(s)) \\
& \stackrel{(c)}{=}  I( \textup{Key} ;  \widetilde{Z}^{\textup{init}}(s)) , \label{eqmotp3}
\end{align}
where $(b)$ holds by the chain rule and because $I(  M_{\textup{OTP}};  \widetilde{Z}^{\textup{init}}(s)| M_{1:B} \widetilde{Z}_{\mathcal{B}}(s))  \leq  I(  M_{\textup{OTP}} M_{1:B} \widetilde{Z}_{\mathcal{B}}(s);  \widetilde{Z}^{\textup{init}}(s))  = 0$, $(c)$ holds by the chain rule and because $I(  M_{\textup{OTP}}M_{1:B} \widetilde{Z}_{\mathcal{B}}(s);  \widetilde{Z}^{\textup{init}}(s) | \textup{Key}) \leq I(  M_{\textup{OTP}}M_{1:B} \widetilde{Z}_{\mathcal{B}}(s);  \widetilde{Z}^{\textup{init}}(s)  \textup{Key}) = 0$. By combining~\eqref{eqmotp1}, \eqref{eqmotp2},  and \eqref{eqmotp3}, we obtain
$
 I(M_{1:B} ;C \widetilde{Z}_{\mathcal{B}}(s) \widetilde{Z}^{\textup{init}}(s)) \leq  I(M_{1:B} ; \widetilde{Z}_{\mathcal{B}}(s) ) + I(\textup{Key}; \widetilde{Z}^{\textup{init}}(s)) + \log |\mathcal{K}| - H(\textup{Key}). 
$ 
 Finally, we obtain the lemma with Lemmas \ref{lemjointsec} and~\ref{lemsecurkeyjoint}.
\end{proof}

\section{Proof of Theorem  \ref{thmain2}}  \label{sec:proof2}
We assume in the following that there exists a best channel for the eavesdropper \cite{molavianjazi2009arbitrary}, i.e., $ \exists s^*\in \overline{\mathfrak{S}}, \forall s \in \mathfrak{S}$, $X-Z (s^*) -Z(s)$. Similar to the proof of Theorem~\ref{thmain1}, we proceed in two steps. We first ignore the initialization phase and assume that the legitimate users have access to a secret key to perform the one-time pad in Algorithms~\ref{alg:p1},~\ref{alg:p}. We only show blockwise security as the remainder of the proof is similar to the proof in Section \ref{sec:proof1a}. We also omit the second step that consists in analyzing the initialization phase jointly with Algorithms~\ref{alg:p1},~\ref{alg:p}, as it is similar to the analysis in Section \ref{sec:initphase}.

\subsection{Blockwise security analysis}
We adopt the same notation as in Section \ref{sec:proof1a}. We have the following inequality, whose proof is identical to the proof of Lemma \ref{lemdist}.
For $b\in \mathcal{B}$, we have
\begin{align}
 \mathbb{D}( q_{U^{1:N}X^{1:N}Y^{1:N}Z^{1:N}(\mathbf{s}_b)} \lVert \widetilde{p}_{U_{b}^{1:N}X_{b}^{1:N}Y_{b}^{1:N}Z_{b}^{1:N}(\mathbf{s}_b)} )
 \leq 2N \delta_K, \label{lemdist2}
\end{align}
where we have defined $ q_{U^{1:N}X^{1:N}Y^{1:N}Z^{1:N}(\mathbf{s}_b)} \triangleq \prod_{i=1}^Nq_{UXYZ(s_{b,i})}$.
Next, similar to Lemma \ref{lemampapp} using \eqref{lemdist2} in place of Lemma \ref{lemdist}, we have for any $\gamma \in]0,1[$
\begin{align}
& \mathbb{V} ( \widetilde{p}_{\bar{M}_{b} R_{b} Z_b(\mathbf{s}_b) {X}_b[\mathcal{A}_b] }, \widetilde{p}_{\bar{M}_{b}} \widetilde{p}_{R_{b}Z_b (\mathbf{s}_b){X}_b[\mathcal{A}_b] }) \nonumber \\
& \leq 2^{1-L^{\gamma}} + 2 \sqrt{ 2^{ r -  {H} \left(T_{b}| \widetilde{Z}_{b} (\mathbf{s}_b) \widetilde{X}_b[\mathcal{A}_b]\right) + N \delta^{(1)}(K,L)  }} , \label{eqam2}
\end{align}
where $\delta^{(1)}(K,L)$ is defined in Lemma \ref{lemampapp}.
We then have
\begin{align}
&{H} \left(T_{b}| \widetilde{Z}_{b} (\mathbf{s}_b) \widetilde{X}_b[\mathcal{A}_b]\right) \nonumber \\ \nonumber
& \stackrel{(a)}{\geq} {H} \left( U_b|{Z}_{b} (\mathbf{s}_b) {X}_b[\mathcal{A}_b]  \right) -  N \delta^{(2)}(K,L)  \displaybreak[0] \\  \nonumber
&\geq {H} \left( U_b| {Z}_{b}(s^*){Z}_{b} (\mathbf{s}_b) {X}_b[\mathcal{A}_b]  \right) -  N \delta^{(2)}(K,L) \displaybreak[0] \\ \nonumber
&\stackrel{(b)}{\geq} {H} \left( U_b| {Z}_{b}(s^*) {X}_b[\mathcal{A}_b]  \right)  -  N \delta^{(2)}(K,L) \displaybreak[0] \\ 
& \stackrel{(c)}{=} N(1-\alpha) H(U|Z(s^*) ) +N\alpha H(U|X)  -  N \delta^{(2)}(K,L),\label{eqam2b}
\end{align}
where $(a)$ holds as in the proof of Lemma \ref{lemcond} with $ \delta^{(2)}(K,L)$ defined in Lemma \ref{lemcond}, $(b)$ holds because $ ({U}_b,{X}_b) - Z_b(s^*) - Z_b(\mathbf{s}_b) $ forms a Markov chain, $(c)$ holds as in the proof of Lemma \ref{lemcond}. Finally, from~\eqref{eqam2} and \eqref{eqam2b}, we can conclude as in Section \ref{sec:secindiv} that blockwise security holds.

\section{Extension to uncertainty on the main channel} \label{secext}
Assume now that uncertainty on the main channel also holds according to a compound model, i.e.,  the channel of Section \ref{sec:main} is now defined by the conditional probabilities $(p_{Y(t)Z(s)|X})_{s\in\mathfrak{S},t\in\mathfrak{T}}$, where $\mathfrak{T}$ is a finite set. Assume also that for all channel uses $s\in\mathfrak{S}$ and $t\in\mathfrak{T}$ are fixed.   We extend Theorem \ref{thmain1} to this setting in Section \ref{sectionproofth8} using new polar coding schemes for source coding with compound side information and for compound channel coding described in Sections \ref{secscci} and \ref{seccompoundcc}, respectively.

\subsection{Source coding with compound side information} \label{secscci}

 \cite{ye2014universal} provides a polar coding scheme with optimal rate for lossless source coding with compound side information. However, for our purposes, we modify the coding scheme in~\cite{ye2014universal} to ensure near uniformity of the encoder output.

Consider a compound source $\left((\mathcal{U} \times \mathcal{Y}_j)_{j \in \mathcal{J}},(p_{UY_j})_{j \in \mathcal{J}} \right)$,  where $\mathcal{U} \triangleq \{ 0,1\}$ and $\mathcal{J} \triangleq \llbracket 1 ,J \rrbracket$. Let $(t_j)_{j \in \mathcal{J}} \in \mathbb{N}^J$ with $t_1 \triangleq 1$ and define for $j \in \mathcal{J}$, $T_j \triangleq  \prod_{i=1}^{j} t_i$ and $N_j \triangleq  N T_j$, where $N$ is a power of two. Consider for $j\in\mathcal{J}$, $(U^{1:N_J},Y_j^{1:N_J})= (U^{1:N}_{t},(Y_j)^{1:N}_{t})_{t \in \llbracket 1 ,T_J  \rrbracket }$ distributed according to the product distribution $p_{U^{1:N_J}Y_j^{1:N_J}}$. For $j_0 \in \mathcal{J}$, we also use the notation $Y_{j_0}^{1:N_{j}} = (Y_{j_0,t}^{1:N_{j-1}})_{t \in \llbracket 1 , t_j \rrbracket}$,  $j\in \llbracket 2 , J \rrbracket$, to indicate that $Y_{j_0}^{1:N_{j}}$ is made of $t_j$ blocks of length $N_{j-1}$.
Define for $t \in \llbracket 1 ,T_J  \rrbracket$, $A_t^{1:N} \triangleq U_t^{1:N} G_N$ and for $\delta_N \triangleq 2^{-N^{\beta}}$, $\beta \in ]0,1/2[$, and $j \in \mathcal{J}$ define the sets
\begin{align*}
\mathcal{V}_U &\triangleq \left\{ i \in \llbracket 1, N \rrbracket : H( A_1^i | A_1^{1:i-1}) >  1 - \delta_N \right\},\\
\mathcal{H}_{U|Y_j} &\triangleq \left\{ i \in \llbracket 1, N \rrbracket : H( A_1^i | A_1^{1:i-1} (Y_j)_1^{1:N}) > \delta_N \right\}, \\
\mathcal{V}_{U|Y_j} &\triangleq \left\{ i \in \llbracket 1, N \rrbracket : H( A_1^i | A_1^{1:i-1} (Y_j)_1^{1:N}) > 1-\delta_N \right\}.
\end{align*}   
We also use the notation $U^{1:N_{j}} = (U_{t}^{1:N_{j-1}})_{t \in \llbracket 1 , t_j \rrbracket}$,  $j \in \llbracket 2 , J  \rrbracket$, to indicate that $U^{1:N_{j}}$ is made of $t_j$ blocks of length $N_{j-1}$.
The encoding is described in Algorithm \ref{alg:scsic}. By the successive cancellation decoder for polar source coding with side information \cite{Arikan10},  Decoder $1$ with  $[e^{(1)} ( U^{1:N_1}),E'_1] = A_1^{1:N} [\mathcal{H}_{U|Y_1}]$ and $Y_1^{1:N}$ can compute a good estimate $\hat{U}_1^{1:N_1}$ of $U^{1:N_1}$. Now, assume that when $L \in \llbracket 1,J-1 \rrbracket$, for any Decoder $l  \in \llbracket 1, L \rrbracket$, there is a function $g^{(L)}_l$ such that $\hat{U}^{1:N_{L}}_{l} \triangleq g^{(L)}_l (e^{(L)} ( U^{1:N_L}),E'_l, Y_l^{1:N_L}) $ is a good estimate of $U^{1:N_L}$. Then, Algorithms \ref{alg:scsicdec} and \ref{alg:scsicdecL} show that any decoder $l \in \llbracket 1 , L+1 \rrbracket$  can form a good estimate $\hat{U}^{1:N_{L+1}}_{l}$ of $U^{1:N_{L+1}}$ from $[e^{(L+1)} ( U^{1:N_{L+1}}),E'_l, Y_l^{1:N_{L+1}}]$.

The encoding and decoding algorithms for source coding with compound side information are described in Algorithms~\ref{alg:scsic}, \ref{alg:scsicdec}, \ref{alg:scsicdecL}, and yield the following result.

\begin{thm}
	The algorithms \ref{alg:scsic}, \ref{alg:scsicdec}, \ref{alg:scsicdecL} perform source coding with compound side information on sequences with length $T_J N$ with optimal rate $\max_{j\in\mathcal{J}}H(U|Y_j)$ and encoding/decoding complexity $T_J N O(\log N)$.
\end{thm}

Note that the encoding is different than in \cite{ye2014universal} as the encoder output is split into $E$ and $E'$, however, the decoder is equivalent to the one in \cite{ye2014universal}. Consequently, the probability of error in the reconstruction of the source asymptotically vanishes by \cite{ye2014universal}. Additionally, remark that the rate of $E'$ is negligible compared to $N_J$ because for any $j \in \mathcal{J}$, $|\mathcal{H}_{U|Y_j} \backslash \mathcal{V}_{U|Y_j}| = |\mathcal{H}_{U|Y_j}|-| \mathcal{V}_{U|Y_j}|= o(N)$ by \cite{Arikan10} and \cite[Lemma 7]{Chou16}. Hence, the coding scheme rate is the same as in \cite{ye2014universal} but now can also be used to ensure a near uniform encoder output  by one-time padding $E'$ with a sequence of $|E'|$ uniformly distributed bits shared by the encoder and decoder.  Note that it generalizes the polar coding schemes for source coding with nearly uniform output \cite{chou2013data} in \cite{chou2015coding,chou2018universal}.  
\begin{algorithm}
  \caption{Encoding}
  \label{alg:scsic}
  \begin{algorithmic}[1] 
\REQUIRE   Assume that the sequence to compress is $U^{1:N_{J}}$
              \STATE Define the function $e^{(1)} : U^{1:N_1} \mapsto  A^{1:N}_1 [ \mathcal{V}_{U|Y_1}]$
              \FOR{$j=1$ to $J-1$}
                  \STATE Define $f^{(j)} : U^{1:N_{j}} \mapsto  (A^{1:N}_t[ \mathcal{V}_{U|Y_{j+1}}])_{t \in \llbracket 1, T_{j} \rrbracket }$
	\STATE				Define the function $e^{(j+1)}$ which maps  $U^{1:N_{j+1}}$ to
	\begin{align*}
	&[e^{(j)}(U^{1:N_{j}}_1),(e^{(j)}(U^{1:N_{j}}_{t+1})\! \oplus \! f^{(j)} (U^{1:N_{j}}_{t} ) )_{t \in \llbracket 1, t_{j+1} -1 \rrbracket },\\
	& \phantom{-}f^{(j)} (U^{1:N_{j}}_{t_{j+1}}  ) ],
		\end{align*}
 (if the two sequences have different lengths, then
the shorter sequence is padded with zeros)
              \ENDFOR  
              \STATE Define $E \triangleq e^{(J)}(U^{1:N_{J}})$
              \STATE For $j \in \mathcal{J}$, define $E'_j \triangleq (A^{1:N}_t[ \mathcal{H}_{U|Y_j} \backslash \mathcal{V}_{U|Y_j}])_{  t \in \llbracket 1, T_{J} \rrbracket } $, and $E' \triangleq (E'_j)_{j \in \mathcal{J}}$. 
              \RETURN $(E,E')$
  \end{algorithmic}  
\end{algorithm}
\begin{algorithm}
  \caption{Decoder $j_0 \in \llbracket 1 , L\rrbracket$}
  \label{alg:scsicdec}
  \begin{algorithmic}[1] 
\REQUIRE  $(E,E')$ and $Y_{j_0}^{1:N_{L+1}}$
	\STATE				Form $\widehat{U}_{j_0,1}^{1:N_{L}} \triangleq g_{j_0}^{(L)}(e^{(L)}(U^{1:N_{L}}_1), E'_{j_0}, Y^{1:N_{L}}_{j_0,1})$ , where $e^{(L)}(U^{1:N_{L}}_1)$ is obtained from $e^{(L+1)}(U^{1:N_{L+1}})$\FOR{Block $t=2$ to Block $t=t_{L+1}$}
	 \STATE Form
	  $\widehat{U}_{j_0,t}^{1:N_{L}} \triangleq g_{j_0}^{(L)}(
	  	  e^{(L)}(U^{1:N_{L}}_t) \oplus f^{(L)} (U^{1:N_{L}}_{t-1} ) \oplus f^{(L)} (\widehat{U}^{1:N_{L}}_{j_0,t-1} ) ,E'_{j_0},Y^{1:N_{L}}_{j_0,t})$
              \ENDFOR  
              \RETURN  $\widehat{U}_{j_0}^{1:N_{L+1}} \triangleq (\widehat{U}_{j_0,t}^{1:N_{L}} )_{t \in \llbracket 1 , t_{L+1} \rrbracket}$, an estimate of ${U}^{1:N_{L+1}}$
  \end{algorithmic}  
\end{algorithm}
\begin{algorithm}
  \caption{Decoder $L+1$}
  \label{alg:scsicdecL}
  \begin{algorithmic}[1] 
\REQUIRE  $(E,E')$ and $Y_{L+1}^{1:N_{L+1}}$
		\STATE			With the successive cancellation decoder for source coding with side information \cite{Arikan10}, form $\widehat{U}_{L+1,t_{L+1}}^{1:N_{L}}$ from 
		$$\left(f^{(L)}(U^{1:N_{L}}_{t_{L+1}}), E'_{L+1} ,Y^{1:N_{L}}_{L+1,t_{L+1}} \right)$$  
		\vspace{-1em}
   \FOR{Block $t=t_{L+1} -1$ to Block $t=1$}
        \STATE Form an estimate $\widehat f^{(L)}  (U^{1:N_{L}}_{t})$ of $f^{(L)} (U^{1:N_{L}}_{t} )$ with   
\begin{multline*}        
        \widehat f^{(L)} (U^{1:N_{L}}_{t} )  \\ \triangleq   f^{(L)} (U^{1:N_{L}}_{t} ) \oplus e^{(L)}(U^{1:N_{L}}_{t+1}) \oplus  e^{(L)}(\widehat{U}^{1:N_{L}}_{L+1,t+1})
        \end{multline*}
        		\vspace{-1em}
     \STATE With the successive cancellation decoder for source coding with side information \cite{Arikan10}, form $\widehat{U}^{1:N_{L}}_{L+1,t}$   from 
     $\left(\widehat f^{(L)} (U^{1:N_{L}}_{t} ),E'_{L+1} , Y^{1:N_{L}}_{L+1,t} \right)$
              \ENDFOR  
              \RETURN  $\widehat{U}^{1:N_{L+1}}_{L+1}  \triangleq (\widehat{U}^{1:N_{L}}_{L+1,t} )_{t \in \llbracket 1 , t_{L+1} \rrbracket}$
  \end{algorithmic}  
\end{algorithm}
 \subsection{Compound channel coding from source coding} \label{seccompoundcc}
 We now propose a capacity-achieving compound channel coding scheme from source coding with compound side information via a technique similar to the one in \cite{Mondelli14} used to obtain channel coding from source coding with side information.

Consider a compound channel $\left(\mathcal{X} ,(p_{Y_j|X})_{j \in \mathcal{J}}, (\mathcal{Y}_j)_{j \in \mathcal{J}} \right)$, where $\mathcal{X} \triangleq  \{ 0,1\}$ and $\mathcal{J} \triangleq \llbracket 1 ,J \rrbracket$.  
Consider an arbitrary distribution $p_X$ on $\mathcal{X}$ and define for $j \in \mathcal{J}$, $p_{XY_j} \triangleq p_X p_{Y_j|X}$.  Consider for $j\in\mathcal{J}$, $(X^{1:N},Y_j^{1:N})$ distributed according to the product distribution $p_{X^{1:N}Y_j^{1:N}}$. Define  $V^{1:N} \triangleq X^{1:N} G_N$ and for $\delta_N \triangleq 2^{-N^{\beta}}$, $\beta \in ]0,1/2[$, and $j \in \mathcal{J}$, define the sets
\begin{align*}
\mathcal{V}_X &\triangleq \left\{ i \in \llbracket 1, N \rrbracket : H( V^i | V^{1:i-1}) >  1 - \delta_N \right\},\\
\mathcal{H}_{X|Y_j} &\triangleq \left\{ i \in \llbracket 1, N \rrbracket : H( V^i | V^{1:i-1} Y_j^{1:N}) > \delta_N \right\}, \\
\mathcal{V}_{X|Y_j} &\triangleq \left\{ i \in \llbracket 1, N \rrbracket : H( V^i | V^{1:i-1} Y_j^{1:N}) > 1-\delta_N \right\}.
\end{align*}   
Let $(t_j)_{j \in \mathcal{J}} \in \mathbb{N}^J$ with $t_1 \triangleq 1$ and define for $j \in \mathcal{J}$, $T_j \triangleq  \prod_{i=1}^{j} t_i$ and $N_j \triangleq  N T_j$. We use the same notation as in Section~\ref{secscci}. Let $|E|$ be the length of the output $E$ in the encoder of source coding with compound side information described in Algorithm \ref{alg:scsic}. By Euclidean division, there exist $q\in \mathbb{N}$ and $r\in \llbracket 1 , T_J-1 \rrbracket$ such that $|E|= T_J q +  r  $. For $t \in \llbracket 1 ,r \rrbracket$, consider an arbitrary set $\mathcal{A}_t \subset \mathcal{V}_{X}$ such that $|\mathcal{A}_t| = q+1 $, and, for $t \in \llbracket r+1 ,T_J  \rrbracket$, consider an arbitrary set $\mathcal{A}_{t} \subset \mathcal{V}_{X}$ such that $|\mathcal{A}_{t}| = q$. Hence, $\sum_{t=1}^{T_J}|\mathcal{A}_t|= |E|$.  
 
 The encoding and decoding algorithms for compound channel coding are described in Algorithms \ref{alg:scsice} and \ref{alg:scsicd}, and yield the following result, whose proof is  similar to \cite{Chou15f}. Note that other capacity-achieving polar coding schemes had also been proposed for compound symmetric channels in \cite{hassani2014universal,csacsouglu2016universal}. 
\begin{thm}
	Algorithms \ref{alg:scsice} and \ref{alg:scsicd} perform compound channel coding over $B$ blocks of length $ T_J N$ with optimal rate $\max_{p_X}\min_{j \in \mathcal{J}} I(X;Y_j)$ and encoding/decoding complexity~$O(B T_J N \log N)$.	
\end{thm}
 \begin{algorithm}
  \caption{Encoder}
  \label{alg:scsice}
  \begin{algorithmic}[1] 
\REQUIRE  $E_0 \triangleq (E_{0,t})_{t \in \llbracket 1 ,T_J  \rrbracket}$, where $E_{0,t}$, $t \in \llbracket 1 ,T_J  \rrbracket$, is a sequence of $|\mathcal{A}_t|$ uniformly distributed bits (local randomness). Messages $(M_{b,t})_{b \in \llbracket 1 , B \rrbracket, t \in \llbracket 1 ,T_J  \rrbracket}$, where $M_{b,t}$, $b \in \llbracket 1 , B \rrbracket, t \in \llbracket 1 ,T_J  \rrbracket$, is a sequence of $|\mathcal{V}_X \backslash \mathcal{A}_t |$  uniformly distributed bits
   \FOR{Block $b=1$ to Block $b=B$}
                          \FOR{Sub-block $t=1$ to Sub-block $t=T_J$}
                                     \STATE Define $\widetilde{V}_{b,t}^{1:N}$  according to $\prod_{j=1}^N \widetilde{p}_{{V}_{b,t}^j|{V}_{b,t}^{1:j-1}}$ with
        \begin{align*} 
&\widetilde{p}_{{V}_{b,t}^j|V_{b,t}^{1:j-1}} (v_{b,t}^j|v_{b,t}^{1:j-1})\\
& \triangleq 
\begin{cases} 
\mathds{1} \{v_{b,t}^j = M_{b,t}^j\} &\text{if }j\in \mathcal{V}_{X} \backslash \mathcal{A}_t  \\
\mathds{1} \{v_{b,t}^j = E_{b-1,t}^j\} &\text{if }j\in   \mathcal{A}_t  \\
 {p}_{V^j|V^{1:j-1}} (v_{b,t}^j|v_{b,t}^{1:j-1}) & \text{if }j\in {\mathcal{V}}_{X}^c  \end{cases} 
  \end{align*}
          \STATE Send $\widetilde X_{b,t}^{1:N} \triangleq \widetilde V_{b,t}^{1:N} G_N$  over the channel. 
                                      \ENDFOR  
                    \STATE Define $(E_b,E'_b)$ as the output of the encoder described in Algorithm \ref{alg:scsic} (for the compound source $(p_{XY_j})_{j\in \mathcal{J}}$)  applied to $\widetilde X_b^{1:N_J} \triangleq (\widetilde X_{b,t}^{1:N})_{t \in \llbracket 1 , T_J \rrbracket} $  
                    \STATE Break down $E_b$ into $T_J$ sequences $(E_{b,t})_{t \in \llbracket 1 ,T_J  \rrbracket}$, such that $|E_{b,t}|= |\mathcal{A}_t|$, $t \in \llbracket 1, T_J \rrbracket$.
                            \ENDFOR  
              \STATE Do a one-time pad with $(E'_b)_{b \in \llbracket 1, B \rrbracket}$ and $E_B$ to ensure uniformity (similar to Algorithm \ref{alg:p1}) and send it to the receiver via  channel codes~\cite{Arikan09} for each $p_{Y_j|X}$, $j \in \mathcal{J}$ 
  \end{algorithmic}  
\end{algorithm}
 \begin{algorithm} 
  \caption{Decoder $j\in\mathcal{J}$}
  \label{alg:scsicd}
  \begin{algorithmic}[1] 
\REQUIRE Channel output $\widetilde{Y}^{1:BN_J}_j$, estimate $\widehat E_{B}$ of $E_B$, and estimate  $(\widehat E'_b)_{b \in \llbracket 1, B \rrbracket}$ of  $(E'_b)_{b \in \llbracket 1, B \rrbracket}$
     \FOR{Block $b=B$ to Block $b=1$}
     \STATE Use $(\widehat E_{b},\widehat E'_b)$ with Decoder $j$ in Algorithms \ref{alg:scsicdec}, \ref{alg:scsicdecL} to create an estimate $\widehat X_b^{1:N_J} \triangleq (\widehat X_{b,t}^{1:N})_{t \in \llbracket 1 , T_J \rrbracket} $  of   $\widetilde X_b^{1:N_J} \triangleq (\widetilde X_{b,t}^{1:N})_{t \in \llbracket 1 , T_J \rrbracket} $. 
     \FOR{Sub-block $t=1$ to Sub-block $t=T_J$}
     \STATE Form  an estimate $\widehat V_{b,t}^{1:N} \triangleq \widehat X_{b,t}^{1:N}G_N$  of $\widetilde V_{b,t}^{1:N}$ 
      \STATE Form an estimate $\widehat  M_{b,t} \triangleq \widehat V_{b,t}^{1:N} [\mathcal{V}_{X} \backslash \mathcal{A}_t] $ of $ M_{b,t}$
      \STATE Form an estimate $\widehat  E_{b-1,t} \triangleq \widehat V_{b,t}^{1:N} [\mathcal{A}_t] $ of $ E_{b-1,t}$
        \ENDFOR 
      \STATE Form $\widehat{E}_{b-1} \triangleq (\widehat{E}_{b-1,t})_{t \in \llbracket 1 , T_J \rrbracket}$ an estimate of ${E}_{b-1}$
  \ENDFOR 
  \RETURN  $(\widehat M_{b,t})_{b\in\mathcal{B}, t \in \llbracket 1 , T_J \rrbracket}$
  \end{algorithmic}  
\end{algorithm}
\begin{rem}
We do not write the dependence of the estimates with respect to $j \in \mathcal{J}$ in Algorithm \ref{alg:scsicd} to simplify notation.
\end{rem}

\subsection{Extension to compound uncertainty on the main channel} \label{sectionproofth8}
Using the preliminary results of Section~\ref{secscci} and \ref{seccompoundcc}, an immediate extension of Theorem \ref{thmain1} is as follows.
\begin{thm} \label{thmainb}
Assume that in the coding scheme of Section~\ref{sec:preshared} the primitive source coding with side information is replaced by source coding with compound side information from Section~\ref{secscci}. Assume also that  instead of channel coding in Lines~10 and 15 of Algorithm \ref{alg:p3} and \ref{alg:p1}, respectively, we use compound channel coding from Section \ref{seccompoundcc}. Then, the following secrecy rate  is achieved
$$\max \! \left[ \min_{t\in \mathfrak{T}} I(U;Y(t)) \!-\! \alpha I(U;X) \!-\! (1\!-\!\alpha) \max_{s\in \mathfrak{S}} I ({U}; {Z}(s) ) \right]^+$$
where the maximum is over random variables $U$ such that $\forall t \in \mathfrak{T}, \forall s \in \mathfrak{S},U - X -(Y(t),Z(s))$, and $|\mathcal{U}| \leq |\mathcal{X}| $. 
\end{thm}

\section{Concluding Remarks} \label{sec:concl}
We constructed  explicit wiretap codes that 
achieve the  best known single-letter achievable rates, previously obtained non-constructively,  when uncertainty holds on the eavesdropper channel  under a   (i) noisy~ blockwise type~II, (ii) compound, or (iii)~arbitrarily varying model. Our construction solely relies on three primitives: source coding with side information, universal hashing, and distribution approximation. We also extended our result to the case where uncertainty holds on the legitimate user channel under a compound model. This extension can thus be applied to the problem of secret sharing from correlated randomness. Specifically, it can directly be  applied to the case of a discrete channel model as in \cite[Section II]{zou2015information}, and adapted to the case of a discrete source model with a single dealer, as in \cite{chou2018secret,chou2021distributed}, for arbitrary access structures. The case of Gaussian channels or sources, e.g.,~\cite{zou2015information,rana2021information}, is, however, more challenging as quantization may be~needed. The case of rate-limited communication for source models is also more challenging as vector quantization is needed and requires other proof techniques \cite{sultana2021}.

We anticipate that our code construction can be generalized to the broadcast channel with confidential messages and the multiple access wiretap channel when uncertainty holds on the eavesdropper's channel according to a compound model, using a distributed version of the leftover hash lemma akin to \cite{chou2021private}. Such results would generalize known constructions based on polar codes, e.g., \cite{Wei14,Chou16,chou2018polar}, that require a seed for strong secrecy and assume perfect knowledge  of the eavesdropper's channel statistics.
An open problem is to  provide explicit coding schemes to handle an arbitrarily varying main channel as, for instance, in the models in \cite{molavianjazi2009arbitrary,bjelakovic2013capacity,notzel2016arbitrarily,chen2022strong}.

\appendices
\section{proof of Lemma \ref{lemdist}} \label{App_lemdist}

Let $b\in\mathcal{B}$ and $l \in \mathcal{L}$. By \eqref{eqDA1}, we have
\begin{align}
 \mathbb{D}(q_{A^{1:K}}\lVert \widetilde{p}_{A_{b,l}^{1:K}} ) 
& \leq K \delta_K, \label{eqdist1}
\end{align}
we can indeed apply \eqref{eqDA1} because the bits $\widetilde{A}_{b,l}^{1:K}[\mathcal{V}_U]$ are uniformly distributed, which is a consequence of the definition of $\widetilde{A}_{b,l}^{1:K}[\mathcal{V}_U]$ in Line 7 of Algorithm \ref{alg:p1} using the fact that the bits $T^{1:|\mathcal{V}_U|L}_b = R_b^{-1} \odot (M_b\lVert M'_b\lVert R_b')$ are uniformly distributed since the bits $(M_b\lVert M'_b\lVert R_b')$ are uniformly distributed.
Next, we have
\begin{align}
& \mathbb{D}( q_{U^{1:K}V^{1:K}} \lVert \widetilde{p}_{U_{b,l}^{1:K}V_{b,l}^{1:K}} ) \displaybreak[0] \nonumber\\
& \stackrel{(a)}{=}  \mathbb{E}_{q_{U^{1:K}}} \mathbb{D}( q_{V^{1:K}|U^{1:K}} \lVert \widetilde{p}_{V_{b,l}^{1:K}|U_{b,l}^{1:K}} ) + \mathbb{D}( q_{U^{1:K}} \lVert \widetilde{p}_{U_{b,l}^{1:K}} ) \nonumber \displaybreak[0]  \\
& \stackrel{(b)}{\leq} \mathbb{E}_{q_{U^{1:K}}} \mathbb{D}( q_{V^{1:K}|U^{1:K}} \lVert \widetilde{p}_{V_{b,l}^{1:K}|U_{b,l}^{1:K}} )  + K \delta_K  \nonumber \displaybreak[0] \\ 
& \stackrel{(c)}  \leq 2 K \delta_K,  \label{eqdist2}
\end{align}
where $(a)$ holds by the chain rule for relative entropy~\cite{Cover91}, $(b)$ holds by \eqref{eqdist1} because $\mathbb{D}( q_{U^{1:K}} \lVert \widetilde{p}_{U_{b,l}^{1:K}} ) =  \mathbb{D}(q_{A^{1:K}}\lVert \widetilde{p}_{A_{b,l}^{1:K}} )$ by invertibility of $G_K$, $(c)$ holds by \eqref{eqCP}. Then,
\begin{align}
& \mathbb{D}( q_{U^{1:N}X^{1:N}Y^{1:N}Z^{1:N}(s)} \lVert  \widetilde{p}_{U_{b}^{1:N}X_{b}^{1:N}Y_{b}^{1:N}Z_{b}^{1:N}(s)} ) \nonumber \displaybreak[0] \displaybreak[0] \\
& \stackrel{(a)}{=} \sum_{l \in\mathcal{L}} \mathbb{D}(q_{U^{1:K}X^{1:K}Y^{1:K}Z^{1:K}(s)} \lVert \widetilde{p}_{U_{b,l}^{1:K}X_{b,l}^{1:K}Y_{b,l}^{1:K}Z_{b,l}^{1:K}(s)} ) \displaybreak[0] \nonumber  \displaybreak[0]\\
& \stackrel{(b)}{=} \smash{\textstyle\sum_{l \in\mathcal{L}}} [ \mathbb{D}(q_{U^{1:K}X^{1:K}} \lVert  \widetilde{p}_{U_{b,l}^{1:K}X_{b,l}^{1:K}}) \nonumber\\
& \phantom{--} + \mathbb{E} [\mathbb{D}( q_{Y^{1:K}Z^{1:K}(s)|U^{1:K}X^{1:K}} \lVert  \widetilde{p}_{Y_{b,l}^{1:K}Z_{b,l}^{1:K}(s)|U_{b,l}^{1:K}X_{b,l}^{1:K}} )]  ] \nonumber \displaybreak[0]  \\
& \stackrel{(c)}{=} \sum_{l \in\mathcal{L}} \mathbb{D}(q_{U^{1:K}X^{1:K}} \lVert  \widetilde{p}_{U_{b,l}^{1:K}X_{b,l}^{1:K}}) \nonumber \\ \nonumber
& \stackrel{(d)}{\leq} \sum_{l \in\mathcal{L}}2K \delta_K = 2LK \delta_K,
\end{align}
where $(a)$ holds because the random variables $(\widetilde{U}_{b,l}^{1:K},\widetilde{X}_{b,l}^{1:K},\widetilde{Y}_{b,l}^{1:K},\widetilde{Z}_{b,l}^{1:K}(s) )$ across the different sub-blocks $l \in \mathcal{L}$ are independent by construction (see Algorithm~\ref{alg:p1} and Remark \ref{rem1a}), 
$(b)$ holds by the chain rule for relative entropy \cite{Cover91} and the expectation is over $q_{U^{1:K}X^{1:K}}$, $(c)$ holds because $ \widetilde{p}_{Y_{b,l}^{1:K}Z_{b,l}^{1:K}(s)|U_{b,l}^{1:K}X_{b,l}^{1:K}} =  \widetilde{p}_{Y_{b,l}^{1:K}Z_{b,l}^{1:K}(s)|X_{b,l}^{1:K}}= q_{Y^{1:K}Z^{1:K}(s)|X^{1:K}}= q_{Y^{1:K}Z^{1:K}(s)|U^{1:K}X^{1:K}}$, $(d)$ holds by \eqref{eqdist2} because $\mathbb{D}(q_{U^{1:K}X^{1:K}} \lVert  \widetilde{p}_{U_{b,l}^{1:K}X_{b,l}^{1:K}}) = \mathbb{D}(q_{U^{1:K}V^{1:K}} \lVert  \widetilde{p}_{U_{b,l}^{1:K}V_{b,l}^{1:K}})$ by invertibility of $G_K$.

\section{Proof of Lemma \ref{lemint}} \label{App_lemimt}
 For any $( \bar{m}_b, t_b , x_b ,z_b(s)  , r_{b}) $, we have 
\begin{align*}
&\widetilde{p}_{\bar{M}_b T_b  X_b Z_b(s)   R_b} ( \bar{m}_b, t_b, x_b  ,z_b(s)   ,r_{b})   \\ 
&  \stackrel{(a)}{=} \widetilde{p}_{X_b Z_b(s)    |T_b  } (x_b, z_b(s)   |t_b  )  \widetilde{p}_{\bar{M}_b}(\bar{m}_b) \widetilde{p}_{R_{B}}(r_{b}) \\
& \phantom{--} \times \textstyle\sum_{r_{b}' } \widetilde{p}_{R_{b}'} (r_b') \widetilde{p}_{T_b  |R_{b}'\bar{M}_{b}R_{b}} (t_b  |r_b' ,\bar{m}_b,r_b) \displaybreak[0]  \\ 
& \stackrel{(b)}{=} \widetilde{p}_{X_b Z_b(s)    |T_b  } (x_b, z_b(s)   |t_b  ) 2^{-r} |\mathcal{R}|^{-1}  \displaybreak[0] \\
& \phantom{--} \times \textstyle\sum_{r_b'}  \frac{\mathds{1} \{ t_b    = r_b^{-1}\odot ( \bar{m}_b \lVert r_b') \}}{ 2^{-r+|\mathcal{V}_U|L}}  \displaybreak[0] \\ 
& = \widetilde{p}_{X_b Z_b(s)    |T_b  } (x_b, z_b(s)   |t_b  ) 2^{-|\mathcal{V}_U|L} |\mathcal{R}|^{-1} \\
& \phantom{--} \times \textstyle \sum_{r_b'}  \mathds{1} \{ r_{b} \odot t_b    = (\bar{m}_b \lVert r_b')  \} \displaybreak[0]  \\ 
&\stackrel{(c)}{=} \widetilde{p}_{X_b Z_b(s)    |T_b  } (x_b, z_b(s)  |t_b  ) 2^{-|\mathcal{V}_U|L} |\mathcal{R}|^{-1} \mathds{1} \{ F(r_b,t_b  )  =  \bar{m}_b \}  \\
& \stackrel{(d)}{=} \widetilde{q}_{\bar{M}_b T_b X_b   Z_b(s)    R_b} (\bar{m}_b,t_b  ,x_b, z_b(s)  , r_{b}), 
\end{align*}
where $(a)$ holds because  $\widetilde{p}_{\bar{M}_{b}T_b  X_b  {Z} _b  R_{b}}= \widetilde{p}_{X_b Z_b |T_b  } \widetilde{p}_{\bar{M}_{b}} \widetilde{p}_{R_{b}} \widetilde{p}_{T_b  |\bar{M}_{b}R_{b}}$ and  $R_{b}'$ is independent of $(\bar{M}_{b},R_{b})$, $(b)$ holds by uniformity of $\bar{M}_{b}$, $R_{b}$, $R_{b}'$, and by definition of $T_b  $, $(c)$ holds because ($F(r_b,t_b  )  =  \bar{m}_b) \implies  (\sum_{r_{b}'} \mathds{1} \{ r_{b} \odot t_b    = (\bar{m}_b, r_b') \}=1)$ (because $\exists! \underline{r}_b' \in \{0,1\}^{|\mathcal{V}_U|L-r}$ such that $r_{b} \odot t_b    = (\bar{m}_b\lVert  \underline{r}_b')$) and ($F(r_b,t_b  )  \neq  \bar{m}_b) \implies  (\sum_{r_{b}'} \mathds{1} \{ r_{b} \odot t_b    = (\bar{m}_b, r_b') \}=0)$, 
 $(d)$ holds by definition of $\widetilde{q}$.

\section{Proof of Lemma \ref{lemampapp}} \label{App_lemampapp}
We have
\begin{align*}
&\mathbb{V} ( \widetilde{p}_{\bar{M}_{b}R_{b}Z_b(s){X}_b[\mathcal{A}_b]  }, \widetilde{p}_{\bar{M}_{b}} \widetilde{p}_{R_{b}Z_b(s) {X}_b[\mathcal{A}_b] }) \displaybreak[0] \\
& \stackrel{(a)}{=} \mathbb{V} ( \widetilde{q}_{F(R_b,T_b  )R_{b}Z_b(s) {X}_b[\mathcal{A}_b]  }, \widetilde{q}_{\bar{M}_{b}} \widetilde{q}_{R_{b}} \widetilde{q}_{Z_b(s) {X}_b[\mathcal{A}_b] })  \displaybreak[0] \\
& \stackrel{(b)}{\leq} 2\epsilon + \sqrt{ 2^{ r - {H}_{\infty}^{\epsilon} \left(\widetilde{p}_{T_{b} {Z}_{b}(s) {X}_b[\mathcal{A}_b]} | \widetilde{p}_{ {Z}_{b}(s) {X}_b[\mathcal{A}_b]} \right)  }} \displaybreak[0] \\
& \stackrel{(c)}{\leq} 2 \cdot 2^{-L^{\gamma}} + \sqrt{ 2^{ r - {H} \left(T_{b} | \widetilde{Z}_{b} (s) \widetilde{X}_b[\mathcal{A}_b]  \right) + L \delta^{(0)}(K,L)  }} \displaybreak[0] \\
& \stackrel{(d)}{\leq} 2 \cdot 2^{-L^{\gamma}} + \sqrt{ 2^{ r - {H} \left(T_{b} | \widetilde{Z}_{b} (s) \widetilde{X}_b[\mathcal{A}_b]  \right) + N \delta^{(1)}(K,L)  }} ,
\end{align*}
where $(a)$ holds by Lemma \ref{lemint} and the definition of $\widetilde{q}$, $(b)$ holds by Lemmas \ref{lemamp} and \ref{lemint}, $(c)$ holds by Lemma \ref{lemholenstein} below, which can indeed be applied by Remark \ref{rem2}, with $\epsilon \triangleq 2^{-L^{\gamma}}$, $\delta^{(0)}(K,L) \triangleq \sqrt{2 L^{\gamma -1}} \log (2^{|\mathcal{V}_U|} +3)$, $(d)$ holds by choosing $\delta^{(1)}(K,L) \triangleq (K^{-1} +1) \sqrt{2 L^{\gamma -1}}  \geq \delta^{(0)}(K,L) /K$.
\begin{lem}[ \cite{holenstein2011randomness} ] \label{lemholenstein}
Let $p_{X^LZ^L} \triangleq \prod_{i=1}^L p_{X_iZ_i}$ be a probability distribution over $\mathcal{X}^L \times \mathcal{Z}^L$. For any $\delta>0$, ${H}_{\infty}^{\epsilon}(p_{X^LZ^L} | p_{Z^L})\geq H(X^L|Z^L)- L\delta$ , where $\epsilon \triangleq 2^{- \frac{L \delta^2} {2 \log^2(|\mathcal{X}|+3)}} $.
\end{lem}
\begin{rem}
An argument similar to \cite[Lemma 10]{Maurer00} to lower bound the min-entropy   would require adding  an extra round of reconciliation to the coding scheme as in \cite{chou2014separation}. Lemma \ref{lemholenstein} appears to be a simpler alternative here.
\end{rem}

\section{Proof of Lemma \ref{lemcond}} \label{App_lemcond}

We first introduce some notation for convenience. Define for any $\mathcal{I} \subseteq \llbracket 1, K \rrbracket$, $\widetilde{A}_{b} [\mathcal{I}] \triangleq (\widetilde{A}_{b,l}^{1:K}[\mathcal{I}])_{l \in \mathcal{L}}$ and $\widetilde{A}_{b} \triangleq (\widetilde{A}_{b,l}^{1:K})_{l \in \mathcal{L}}$. For $b \in \mathcal{B}$, consider $({U}_{b,l}^{1:K},{X}_{b,l}^{1:K},{Z}_{b,l}^{1:K}(s))_{l \in \mathcal{L}}$ distributed according to $q_{U^{1:N} X^{1:N} Z^{1:N}(s)} \triangleq \prod_{i=1}^{N} q_{U X Z(s)} $ and define for $l \in \mathcal{L}$, ${A}_{b,l}^{1:K} \triangleq {U}_{b,l}^{1:K} G_K $. Next, define for any $\mathcal{I} \subseteq \llbracket 1, K \rrbracket$, ${A}_{b} [\mathcal{I}] \triangleq ({A}_{b,l}^{1:K}[\mathcal{I}])_{l \in \mathcal{L}}$ and ${A}_{b} \triangleq ({A}_{b,l}^{1:K})_{l \in \mathcal{L}}$. Define ${U}_{b} [\mathcal{A}_b] \triangleq ({U}_{b,l}^{1:K}[\mathcal{A}_{b,l}])_{l \in \mathcal{L}}$, ${U}_{b} [\mathcal{A}_b^c] \triangleq ({U}_{b,l}^{1:K}[\mathcal{A}_{b,l}^c] )_{l \in \mathcal{L}}$, ${U}_{b} \triangleq ({U}_{b,l}^{1:K})_{l \in \mathcal{L}}$, ${X}_{b} [\mathcal{A}_b] \triangleq ({X}_{b,l}^{1:K}[\mathcal{A}_{b,l}])_{l \in \mathcal{L}}$, ${X}_{b} [\mathcal{A}_b^c] \triangleq ({X}_{b,l}^{1:K}[\mathcal{A}_{b,l}^c] )_{l \in \mathcal{L}}$, ${X}_{b} \triangleq ({X}_{b,l}^{1:K})_{l \in \mathcal{L}}$, ${Z}_{b}(s) [\mathcal{A}_b] \triangleq ({Z}_{b,l}^{1:K}(s)[\mathcal{A}_{b,l}])_{l \in \mathcal{L}}$, ${Z}_{b}(s) [\mathcal{A}_b^c] \triangleq ({Z}_{b,l}^{1:K}(s)[\mathcal{A}_{b,l}^c] )_{l \in \mathcal{L}}$, ${Z}_{b}(s) \triangleq ({Z}_{b,l}^{1:K}(s))_{l \in \mathcal{L}}$. 
Then, we have
\begin{align}
& {H} (\widetilde{A}_{b} [\mathcal{V}_U]| \widetilde{Z}_{b}(s) \widetilde{X}_b[\mathcal{A}_b]  ) - {H} ({A}_{b} [\mathcal{V}_U]| {Z}_{b}(s) {X}_b[\mathcal{A}_b]  )   \nonumber  \\
& =  {H} (\widetilde{A}_{b} [\mathcal{V}_U] \widetilde{Z}_{b}(s) \widetilde{X}_b[\mathcal{A}_b]  ) -  {H} ({A}_{b} [\mathcal{V}_U] {Z}_{b}(s) {X}_b[\mathcal{A}_b]  ) \nonumber \\ 
& \phantom{--} + {H} ( {Z}_{b}(s) {X}_b[\mathcal{A}_b]  ) - {H} ( \widetilde{Z}_{b}(s) \widetilde{X}_b[\mathcal{A}_b]  )  \displaybreak[0]\nonumber    
\\ \nonumber
& \geq - 2  \sqrt{2 \ln2} \sqrt{2LK \delta_K} \log \frac{(|\mathcal{X}|^2 |\mathcal{Z}_s|)^{LK}}{\sqrt{2 \ln2} \sqrt{2LK \delta_K}} \nonumber   \displaybreak[0] \\ \nonumber   
& \geq -  2  \sqrt{2 \ln2} \sqrt{2LK \delta_K} ( LK \log (|\mathcal{X}|^2 \max_{s \in \mathfrak{S}}|\mathcal{Z}_s| )  \nonumber \displaybreak[0]\\
& \phantom{--}  - \log (\sqrt{2 \ln2} \sqrt{2LK \delta_K}))  \triangleq - \delta^{*}, \label{eqd2}
\end{align}
where the first inequality holds by \cite[Lemma~2.7]{bookCsizar} applied twice because for $N$ large enough, $\mathbb{V}( q_{{A}_{b} [\mathcal{V}_U] {Z}_{b}(s) {X}_b[\mathcal{A}_b]} , \widetilde{p}_{{A}_{b} [\mathcal{V}_U] {Z}_{b}(s) {X}_b[\mathcal{A}_b]} ) \leq \sqrt{2 \ln2} \sqrt{ \mathbb{D}( q_{{A}_{b} [\mathcal{V}_U] {Z}_{b}(s) {X}_b[\mathcal{A}_b]} \lVert \widetilde{p}_{{A}_{b} [\mathcal{V}_U] {Z}_{b}(s) {X}_b[\mathcal{A}_b]} )} \leq   \sqrt{2 \ln2} \sqrt{\mathbb{D}( q_{U^{1:N}X^{1:N}Y^{1:N}Z^{1:N}(s)} \lVert \widetilde{p}_{U_{b}^{1:N}X_{b}^{1:N}Y_{b}^{1:N}Z_{b}^{1:N}(s)} )} \leq \sqrt{2 \ln2} \sqrt{2LK \delta_K}$
 where we have used Pinsker's inequality, the chain rule for divergence, positivity of the divergence, and Lemma~\ref{lemdist}. Then, we have
\begin{align}
&{H} \left(T_{b}| \widetilde{Z}_{b}(s) \widetilde{X}_b[\mathcal{A}_b]  \right) \nonumber   \\ \nonumber
& \stackrel{(a)}{=} {H} \left(\widetilde{A}_{b} [\mathcal{V}_U]| \widetilde{Z}_{b}(s) \widetilde{X}_b[\mathcal{A}_b]  \right) \displaybreak[0] \\ \nonumber
& \stackrel{(b)}{\geq} {H} \left({A}_{b} [\mathcal{V}_U]| {Z}_{b}(s) {X}_b[\mathcal{A}_b]  \right) - \delta^{*} \displaybreak[0]\\ \nonumber
& = {H} \left({A}_{b} [\mathcal{H}_U]| {Z}_{b}(s) {X}_b[\mathcal{A}_b]  \right)\nonumber \displaybreak[0]\\ \nonumber
& \phantom{--} - {H} \left({A}_{b} [\mathcal{H}_U\backslash \mathcal{V}_U]| {A}_{b} [\mathcal{V}_U] {Z}_{b}(s) {X}_b[\mathcal{A}_b]  \right) - \delta^{*} \displaybreak[0]\\ \nonumber
& \geq {H} \left({A}_{b} [\mathcal{H}_U]| {Z}_{b}(s) {X}_b[\mathcal{A}_b]  \right) - L|\mathcal{H}_U\backslash \mathcal{V}_U| - \delta^{*} \displaybreak[0] \\ \nonumber
& \stackrel{(c)}{=} {H} \left({A}_{b} [\mathcal{H}_U]| {Z}_{b}(s) {X}_b[\mathcal{A}_b]  \right) - o(LK) - \delta^{*}\\ \nonumber
&= {H} \left({A}_{b} [\mathcal{H}_U] U_b| {Z}_{b}(s) {X}_b[\mathcal{A}_b]  \right) \nonumber \\ \nonumber
& \phantom{--}  - {H} \left( U_b| {A}_{b} [\mathcal{H}_U] {Z}_{b}(s) {X}_b[\mathcal{A}_b]  \right) - o(LK) - \delta^{*}\\ \nonumber
& \stackrel{(d)}{\geq} {H} \left({A}_{b} [\mathcal{H}_U] U_b| {Z}_{b}(s) {X}_b[\mathcal{A}_b]  \right) \nonumber \\ \nonumber
& \phantom{--}  -  {H}_b \left( LK \delta_K  \right) - (L K)^2 \delta_K - o(LK) - \delta^{*}  \displaybreak[0] \\ \nonumber
& \geq {H} \left( U_b| {Z}_{b}(s) {X}_b[\mathcal{A}_b]  \right) \nonumber \\ \nonumber
& \phantom{--}  -  {H}_b \left( LK \delta_K  \right) - (L K)^2 \delta_K - o(LK) - \delta^{*} \displaybreak[0] \\ 
&\stackrel{(e)}{=} {H} \left( U_b| {Z}_{b}(s)[\mathcal{A}_b^c] {X}_b[\mathcal{A}_b]  \right) \nonumber \\ 
& \phantom{--}  -  {H}_b \left( LK \delta_K  \right) - (L K)^2 \delta_K - o(LK) - \delta^{*}, \label{eqenta1}
\end{align}
where $(a)$ holds by definition of $\widetilde{A}_{b} [\mathcal{V}_U]$, $(b)$ holds by \eqref{eqd2}, $(c)$ holds because $\lim_{K \to \infty} |\mathcal{H}_{U}| / K = H(U)$ by~\cite{Arikan10}, and $\lim_{K \to \infty} |\mathcal{V}_{U}| / K = H(U)$ by \cite{Honda13,Chou14rev}, $(d)$ holds by Fano's inequality since the error probability in the reconstruction of $U_b$ from $A_b[\mathcal{H}_U]$ is upper-bounded by $LK\delta_K$ by the result for source coding with side information from \cite{Arikan10}, reviewed in \eqref{lemscs}, and the union bound, $(e)$ holds because $U_b - ({Z}_{b}(s)[\mathcal{A}_b^c], {X}_b[\mathcal{A}_b])- {Z}_{b}(s)[\mathcal{A}_b]$ forms a Markov chain. Next, we have
\begin{align}
&H ( U_b| {Z}_{b}(s)[\mathcal{A}_b^c] {X}_b[\mathcal{A}_b]  ) \nonumber \displaybreak[0]\\ \nonumber
& = H ({U}_{b}[\mathcal{A}_b^c] | {Z}_{b}(s)[\mathcal{A}_b^c]  {X}_b[\mathcal{A}_b] ) \\
& \phantom{--}+ {H} ({U}_{b}[\mathcal{A}_b] | {U}_{b}[\mathcal{A}_b^c] {Z}_{b}(s)[\mathcal{A}_b^c] {X}_b [\mathcal{A}_b]  ) \nonumber\displaybreak[0]  \\ \nonumber
&\stackrel{(a)}{=} H({U}_{b}[\mathcal{A}_b^c] | {Z}_{b}(s)[\mathcal{A}_b^c]  ) \nonumber\\
& \phantom{--}+ H ({U}_{b}[\mathcal{A}_b]  {X}_b [\mathcal{A}_b] | {U}_{b}[\mathcal{A}_b^c] {Z}_{b}(s)[\mathcal{A}_b^c]    )\nonumber\\ \nonumber
& \phantom{--}-  {H} ( {X}_b [\mathcal{A}_b] |  {U}_{b}[\mathcal{A}_b^c] {Z}_{b}(s)[\mathcal{A}_b^c]  ) \displaybreak[0] \\ \nonumber
&\stackrel{(b)}{=} H({U}_{b}[\mathcal{A}_b^c] | {Z}_{b}(s)[\mathcal{A}_b^c]  ) + H ({U}_{b}[\mathcal{A}_b]  {X}_b [\mathcal{A}_b]   )  -  H ( {X}_b [\mathcal{A}_b]   ) \\ \nonumber
& = H ({U}_{b}[\mathcal{A}_b^c] | {Z}_{b}(s)[\mathcal{A}_b^c]  ) + H ({U}_{b}[\mathcal{A}_b]  | {X}_b [\mathcal{A}_b]   ) \\
&\stackrel{(c)}{=} N(1-\alpha) H(U|Z(s) ) +N\alpha H(U|X)  \label{eqenta2},
\end{align}
where  $(a)$ holds because ${X}_b[\mathcal{A}_b]$ is independent of $({U}_{b}[\mathcal{A}_b^c] , {Z}_{b}(s)[\mathcal{A}_b^c])$, $(b)$ holds because $({U}_{b}[\mathcal{A}_b],  {X}_b [\mathcal{A}_b])$ is independent of $({U}_{b}[\mathcal{A}_b^c], {Z}_{b}(s)[\mathcal{A}_b^c])$ and ${X}_b [\mathcal{A}_b]$ is independent of$({U}_{b}[\mathcal{A}_b^c], {Z}_{b}(s)[\mathcal{A}_b^c])$, $(c)$ holds because $q_{U^{1:N} X^{1:N} Z^{1:N}(s)} = \prod_{i=1}^{N} q_{U X Z(s)} $.
We obtain the lemma from \eqref{eqenta1} and \eqref{eqenta2}.

\bibliographystyle{IEEEtran}
\bibliography{polarwiretap}

\end{document}